%% file: main.tex
\documentclass[11pt]{article}
\usepackage{lmodern} % version of computer modern that supports T1 encoding
\usepackage{cmap} % makes PDFs consistently copy pastable
\usepackage[pdfstartview=FitH,pdfpagemode=UseNone,colorlinks,linkcolor=blue,filecolor=blue,citecolor=red,urlcolor=red]{hyperref}
\usepackage{fullpage}
\usepackage{enumitem}
\usepackage{comment}
\usepackage{multirow}
\usepackage{float}
\usepackage{amsmath,amsthm,amssymb,xcolor,mathtools}
\usepackage[capitalize]{cleveref}
\usepackage{xspace}  % define \newcommands that don't eat spaces
\input{macros}

\begin{document}
\title{Parallel Repetition for the GHZ Game: A Simpler Proof}
\author{Uma Girish\thanks{ Department of Computer Science, Princeton University. E-mail: \href{ugirish@cs.princeton.edu}{\texttt{ugirish@cs.princeton.edu}.} Research supported by the Simons Collaboration on Algorithms and Geometry, by a Simons Investigator Award and by the National Science Foundation grants No. CCF-1714779, CCF-2007462.} \and Justin Holmgren\thanks{NTT Research.  E-mail: \href{justin.holmgren@ntt-research.com}{\texttt{justin.holmgren@ntt-research.com}.}} \and Kunal Mittal\thanks{ Department of Computer Science, Princeton University. E-mail: \href{kmittal@cs.princeton.edu}{\texttt{kmittal@cs.princeton.edu}.} Research supported by the Simons Collaboration on Algorithms and Geometry, by a Simons Investigator Award and by the National Science Foundation grants No. CCF-1714779, CCF-2007462.} \and Ran Raz\thanks{ Department of Computer Science, Princeton University.  E-mail: \href{ranr@cs.princeton.edu}{\texttt{ranr@cs.princeton.edu}.}  Research supported by the Simons Collaboration on Algorithms and Geometry, by a Simons Investigator Award and by the National Science Foundation grants No. CCF-1714779, CCF-2007462. } \and Wei Zhan\thanks{ Department of Computer Science, Princeton University.  E-mail: \href{weizhan@cs.princeton.edu}{\texttt{weizhan@cs.princeton.edu}.} Research supported by the Simons Collaboration on Algorithms and Geometry, by a Simons Investigator Award and by the National Science Foundation grants No. CCF-1714779, CCF-2007462.}}
\date{}
\maketitle

\begin{abstract}
We give a new proof of the fact that the parallel repetition of the (3-player) GHZ game reduces the value of the game to zero polynomially quickly. That is, we show that the value of the $n$-fold $\GHZ$ game is at most $n^{-\Omega(1)}$. This was first established by Holmgren and Raz~\cite{HR20}. We present a new proof of this theorem that we believe to be simpler and more direct. 
Unlike most previous works on parallel repetition, our proof makes no use of information theory, and relies on the use of Fourier analysis.

The GHZ game~\cite{GHZ} has played a foundational role in the understanding of quantum information theory, due in part to the fact that quantum strategies can win the GHZ game with probability $1$.
It is possible that improved parallel repetition bounds may find applications in this setting. 

Recently, Dinur, Harsha, Venkat, and Yuen~\cite{DHVY17} highlighted the GHZ game as a simple three-player game, which is in some sense maximally far from the class of multi-player games whose behavior under parallel repetition is well understood.
Dinur et al. conjectured that parallel repetition decreases the value of the GHZ game exponentially quickly, and speculated that progress on proving this would shed light on parallel repetition for general multi-player (multi-prover) games.
\end{abstract}

\clearpage

\input{intro_2}
\input{prelims_2}

\section{Partitioning into Pseudorandom Subspaces}
We make use of the notion of \textdef{affine partition} similar to the one defined in~\cite{HR20}.
We say that $\Pi$ is an \textdef{affine partition of $(\F_2^n)^3$ of codimension at most $d$} if $\Pi$ is a partition on $(\F_2^n)^3$ and:
\begin{itemize}
\item Each part $\pi\in\Pi$ has the form $a_\pi + \cV_\pi^3$ where
  $\cV_\pi$ is a subspace of $\F_2^n$ and $a_\pi\in(\F_2^n)^3$, and
\item Each $\cV_\pi$ has codimension at most $d$.
\end{itemize}
The main take-away from this section is \cref{cor:pseudorandom}, which states the following: Given the query distribution to the $n$-fold $\GHZ$ game, %\unote{changed slightly}%
and a product event $E\subseteq(\F_2^n)^3$ with large enough probability mass, we can find an affine partition $\Pi$ of $(\F_2^n)^3$ such that on a typical part $\pi\in\Pi$, the non-zero Fourier coefficients of the indicator functions $E_1|_{\pi_1},E_2|_{\pi_2},E_3|_{\pi_3}$ are small. Recall that $E_i|_{\pi_i}:\pi_i\to\{0,1\}$ is the indicator function of the set $E_i\cap\pi_i \subseteq \pi_i$. 

Formally, the proposition is as follows:

\begin{proposition}\label{cor:pseudorandom}
Let $\cP=\cQ^n$. Let $E=E_1\times E_2\times E_3 \subseteq (\F_2^n)^3$ be such that $\cP(E) = \alpha$. For all $\delta>0$, there exists an affine partition $\Pi$ of $(\bbF_2^n)^3$ of codimension at most $\frac{3}{\delta^3}$ such that the following holds. With probability at least $1-\frac{\delta}{\alpha}$ over $\pi\sim \Pi(\cP|E)$, for all $i\in[3]$  and non-zero $\chi \in \widehat{\cV}$, we have $ \abs{\widehat{E_i|_{ \pi_i}}(\chi)}\le   \delta$, where $\pi$ is of the form $\pi_1\times\pi_2\times\pi_3$ for affine shifts $\pi_1,\pi_2,\pi_3$ of some subspace $\cV$ of $\F_2^n$.

\end{proposition}

Recall that $\Pi(\cP|E)$ is the distribution induced by sampling $x\sim \cP|E$ and outputting the part of $\Pi$ to which $x$ belongs. %\unote{Changed slightly.}
%Recall that $\Pi(\cP|E)$ is the random part of $\Pi$ where $x\sim \cP|E$ belongs
Note that in the statement of the proposition, we don't specify a choice of Fourier basis for $\pi_i$. This is because for any set $S\subseteq \pi_i$, the quantity $\abs{\widehat{S}(\chi_{a_i})}$ is independent of choice of $a_i\in \pi_i$ so we simply write $\abs{\widehat{S}(\chi)}$. The proof of \cref{cor:pseudorandom} is similar in nature to the proof of Lemma 6.2 in \cite{HR20}, but is much simpler and is deferred to the \cref{app:pseudorandom}.

\section{Key Fourier Analytic Lemmas}
	 
We crucially make use of the following lemma. %\wnote{might be better to recall what is $\cQ$}

\begin{lemma}\label{lem:fourierlemmaedge}
	Let $\cV\subseteq \bbF_2^n$ be a subspace and $a_1,a_2,a_3\in \bbF_2^n$ be such that $a_1 + a_2 + a_3 = 0$. Let $\pi=\pi_1\times \pi_2\times \pi_3$ where $\pi_i=a_i+\cV$. Let $A\subseteq \pi_1,B\subseteq\pi_2,C\subseteq \pi_3$ be sets such that for all non-zero $\chi\in \widehat{\cV}$, we have $   \abs{\widehat{C}(\chi)}\le \delta_1$. Then,
	\[ \abs{ \underset{\substack{z\sim {\pi_3}\\x\sim \pi_1}}{\E}[A(x)\cdot B(x+z)\cdot C(z)] -  \mu_{\pi_1}(A)\cdot \mu_{\pi_2}(B)\cdot \mu_{\pi_3}(C)} \le \delta_1 . \]
	If furthermore for all non-zero $\chi\in \widehat{\cV}$, we have $\abs{\widehat{B}(\chi)} \le \delta_2  $, then 
	\[ \abs{ \E_{z\sim {\pi_3}}\sbra{\pbra{ \E_{x\sim \pi_1}[A(x)\cdot B(x+z)]}^2\cdot C(z)} -  \mu_{\pi_1}(A)^2\cdot \mu_{\pi_2}(B)^2\cdot \mu_{\pi_3}(C)} \le\delta_2^2+ \delta_1 . \]

\end{lemma}
Recall from \Cref{sec:measure} that $\mu_{\pi_i}(S)\triangleq \tfrac{|S\cap \pi_i|}{|\pi_i|}$. In the statement of this lemma, we don't specify a choice of Fourier basis for $\pi_2$ and $\pi_3$.  Since the properties $\abs{\widehat{C}(\chi_{a_3})}\le \delta_1$ and $\abs{\widehat{B}(\chi_{a_2})}\le \delta_2$ are independent of the choice of $a_2$ and $a_3$, we simply write $\abs{\widehat{C}(\chi)}\le \delta_1$ and $\abs{\widehat{B}(\chi)}\le \delta_2$.

 \begin{proof}[Proof of \cref{lem:fourierlemmaedge}] 
 	 We expand the indicator functions of the sets $A,B,C$ in their Fourier basis with respect to $a_1,a_2,a_3$ as follows.
 	 \begin{align*}
 	  &\E_{\substack{z\sim {\pi_3}\\x\sim \pi_1}}[A(x)\cdot B(x+z)\cdot C(z)]  \\
 	  &= \sum_{\chi,\chi',\chi''\in \widehat{\cV}} \E_{\substack{z\sim {\pi_3}\\x\sim \pi_1}}\sbra{\widehat{A}(\chi_{a_1})\cdot \chi_{a_1}(x)\cdot  \widehat{B}(\chi'_{a_2})\cdot \chi'_{a_2}(x+z)\cdot  \widehat{C}(\chi''_{a_3})\cdot \chi''_{a_3}(z)} \\
 	  &= \sum_{\chi,\chi',\chi''\in \widehat{\cV}} \E_{\substack{z\sim {\pi_3}\\x\sim \pi_1}}\sbra{\widehat{A}(\chi_{a_1})\cdot\widehat{B}(\chi'_{a_2})\cdot \widehat{C}(\chi''_{a_3})\cdot \chi(x+a_1)\cdot   \chi'(x+z+a_2)\cdot   \chi''(z+a_3) }\\
 	  &= \sum_{\chi,\chi',\chi''\in \widehat{\cV}} \E_{x',z'\sim \cV}\sbra{\widehat{A}(\chi_{a_1})\cdot\widehat{B}(\chi'_{a_2})\cdot \widehat{C}(\chi''_{a_3})\cdot \chi(x')\cdot   \chi'(x'+z')\cdot   \chi''(z') }
 	 \end{align*}
   The second equality follows from definition. The third follows from the fact that when we vary over $x\sim a_1+\cV$ and $z\sim a_3+\cV$, the distribution of $(x+a_1,x+z+a_2,z+a_3)$ is the uniform distribution over $\cV^3\cap \{(x,x+z,z)\, | \, x,z\in \F_2^n\}$. Since the Fourier characters form an orthonormal basis, when we take an expectation over $x',z'\sim \cV$, the only terms that survive correspond to $\chi=\chi'=\chi''$. This implies that
    \begin{align*}
   	\E_{\substack{z\sim {\pi_3}\\x\sim \pi_1}}[A(x)\cdot B(x+z)\cdot C(z)]  &= \sum_{\chi\in \widehat{\cV}} \sbra{\widehat{A}(\chi_{a_1})\cdot\widehat{B}(\chi_{a_2})\cdot \widehat{C}(\chi_{a_3}) }\\
   	&=\mu_{\pi_1}(A)\cdot\mu_{\pi_2}(B)\cdot\mu_{\pi_3}(C)+\sum_{\emptyset\neq \chi\in \widehat{\cV}} \sbra{\widehat{A}(\chi_{a_1})\cdot\widehat{B}(\chi_{a_2})\cdot \widehat{C}(\chi_{a_3}) }.
   \end{align*}
The last equality is because $\mu_{\pi_i}(S)=\widehat{S}(\emptyset)$ for any $S\subseteq\pi_i$.  Thus, we have
\[  \abs{ \E_{\substack{z\sim {\pi_3}\\x\sim \pi_1}}[A(x)\cdot B(x+z)\cdot C(z)] -  \mu_{\pi_1}(A)\mu_{\pi_2}(B)\mu_{\pi_3}(C)}\le  \sum_{\chi\in\widehat{\cV}\setminus \emptyset} \abs{\widehat{A}(\chi_{a_1})\cdot \widehat{B}(\chi_{a_2})\cdot \widehat{C}(\chi_{a_3})}. \]
Since all Fourier bases differ only up to a sign, the R.H.S. of the above is independent of choice of $a_1,a_2,a_3$. Henceforth, we omit the subscript. Thus, 
 	\begin{align*}
 		  &\abs{ \E_{\substack{z\sim {\pi_3}\\x\sim \pi_1}}[A(x)\cdot B(x+z)\cdot C(z)] -  \mu_{\pi_1}(A)\mu_{\pi_2}(B)\mu_{\pi_3}(C)}\\
 		  &\le  \sum_{\chi\in\widehat{\cV}\setminus \emptyset} \abs{\widehat{A}(\chi)}\cdot \abs{\widehat{B}(\chi)}\cdot \abs{\widehat{C}(\chi)}\\
 		&\le \sum_{\chi\in\widehat{\cV}\setminus \emptyset} \abs{\widehat{A}(\chi)\cdot \widehat{B}(\chi)}\cdot \delta_1\\
 		&\le \sqrt{ \sum_{\chi\in\widehat{\cV}\setminus\emptyset} \widehat{A}(\chi)^2}\cdot  \sqrt{ \sum_{\chi\in\widehat{\cV}\setminus\emptyset} \widehat{B}(\chi)^2}\cdot \delta_1\le \delta_1.
 	\end{align*}
 	The second-last inequality is due to the Cauchy-Schwarz Inequality. The last inequality follows from the fact that $B$ and $A$ are $\{0,1\}$-indicator functions and by Parseval's theorem, their sum of squares of Fourier coefficients is at most 1. This completes the proof of the first statement in \Cref{lem:fourierlemmaedge}.
 	
 	For the second statement, we again expand the indicator functions of the sets $A,B,C$ in the Fourier basis with respect to $a_1,a_2,a_3$ as follows.  
 	\begin{align*} \E_{x\sim \pi_1}\sbra{A(x)\cdot B(x+z)}&=\sum_{\chi,\chi'\in \widehat{\cV}} \E_{x\sim \pi_1}\sbra{ \widehat{A}(\chi_{a_1})\cdot \widehat{B}(\chi'_{a_2})\cdot \chi_{a_1}(x)\cdot\chi'_{a_2}(x+z)}\\
 	&=\sum_{\chi,\chi'\in \widehat{\cV}} \E_{x\sim \pi_1}\sbra{ \widehat{A}(\chi_{a_1})\cdot \widehat{B}(\chi'_{a_2})\cdot \chi(a_1+x)\cdot\chi'(a_1+a_3+x+z)}\\
 	&=\sum_{\chi,\chi'\in \widehat{\cV}} \E_{x\sim \pi_1}\sbra{ \widehat{A}(\chi_{a_1})\cdot \widehat{B}(\chi'_{a_2})\cdot \chi(a_1+x)\cdot\chi'(a_1+x)\cdot \chi'(a_3+z)}
 	\end{align*}
Here, we used the facts that $a_1+a_2+a_3=0$ and that $\chi'$ is a group homomorphism. As we vary $x\sim a_1+\cV$, the distribution of $a_1+x$ is uniform over $\cV$. Since the Fourier characters are orthonormal, the only terms that survive in the above expression correspond to $\chi=\chi'$. Thus,
 	\begin{align*} \E_{x\sim \pi_1}\sbra{A(x)\cdot B(x+z)}
	&=\sum_{\chi\in \widehat{\cV}}   \widehat{A}(\chi_{a_1})\cdot \widehat{B}(\chi_{a_2}) \cdot \chi(a_3+z).
\end{align*}
We now consider: 
\begin{align*} 
	(*)&:=\E_{z\sim \pi_3}\sbra{ \pbra{\E_{x\sim \pi_1}\sbra{A(x)\cdot B(x+z)}}^2\cdot C(z)}\\ &=	\E_{z\sim \pi_3}\sbra{  \pbra{\sum_{\chi\in \widehat{\cV}}   \widehat{A}(\chi_{a_1})\cdot \widehat{B}(\chi_{a_2}) \cdot \chi(a_3+z)}^2 \cdot C(z)}\\
	&=	\E_{z\sim \pi_3}\sbra{  \sum_{\chi,\chi',\chi''\in \widehat{\cV}}   \widehat{A}(\chi_{a_1})\cdot \widehat{B}(\chi_{a_2}) \cdot \widehat{A}(\chi'_{a_1})\cdot \widehat{B}(\chi'_{a_2}) \cdot \widehat{C}(\chi''_{a_3}) \cdot \chi(a_3+z)\cdot \chi'(a_3+z) \cdot \chi''(a_3+z)}\\
\end{align*}
Recall that there is a canonical group isomorphism between $\cV$ and $\widehat{\cV}$. Under this isomorphism, the $\chi+\chi'+\chi''\in \widehat{\cV}$ satisfies $(\chi+\chi'+\chi'')(a_3+z)=\chi(a_3+z)\cdot \chi'(a_3+z) \cdot \chi''(a_3+z)$ for all $z\in a_3+\cV$. Since the characters form an orthonormal basis, the only terms that survive in the above expression correspond to $\chi+\chi'+\chi''=0.$ This implies that

\[
	(*)=  \sum_{\chi,\chi'\in \widehat{\cV}}   \widehat{A}(\chi_{a_1})\cdot \widehat{B}(\chi_{a_2}) \cdot \widehat{A}(\chi'_{a_1})\cdot \widehat{B}(\chi'_{a_2}) \cdot \widehat{C}((\chi+\chi')_{a_3})\]
Recall that $\widehat{S}(\emptyset)=\mu_{\pi_i}(S)$ for any set $S\subseteq \pi_i$ and $i\in[3]$. Thus,
 	\begin{align*}
 	&\abs{ \E_{z\sim {\pi_3}}\sbra{\pbra{ \E_{x\sim \pi_1}[A(x)\cdot B(x+z)]}^2\cdot C(z)} -  \mu_{\pi_1}(A)^2\cdot \mu_{\pi_2}(B)^2\cdot \mu_{\pi_3}(C)} \\
	&\le \sum_{\emptyset\neq \chi\in\widehat{\cV}} \widehat{A}(\chi)^2\cdot \widehat{B}(\chi)^2\cdot \abs{\widehat{C}(\emptyset)} + \sum_{\chi\neq \chi'\in \widehat{\cV}} \abs{\widehat{A}(\chi)\cdot\widehat{A}(\chi')\cdot \widehat{B}(\chi)\cdot\widehat{B}(\chi')\cdot \widehat{C}(\chi+\chi')}.
\end{align*} 
 We omitted the subscripts on Fourier characters, as the R.H.S. of the above is independent of choice of Fourier basis.
 We now bound the first term by $\delta_2^2$ as follows. We bound $\abs{\widehat{C}(\emptyset)}$ by 1. The given assumption that  $\abs{\widehat{B}(\chi)} \le \delta_2$ for all $\chi\in \widehat{\cV}\setminus\emptyset$ implies that
\[ \sum_{\emptyset\neq \chi\in\widehat{\cV}} \widehat{A}(\chi)^2 \cdot \widehat{B}(\chi)^2 \cdot\abs{ \widehat{C}(\emptyset)} \le \sum_{\emptyset\neq \chi\in\widehat{\cV}}   \widehat{A}(\chi)^2 \cdot{\delta_2}^2\le {\delta_2}^2. \]
We bound the second term as follows.
\begin{align*}
\sum_{\chi\neq \chi'\in \widehat{\cV}}\abs{ \widehat{A}(\chi) \cdot \widehat{A}(\chi') \cdot \widehat{B} (\chi) \cdot \widehat{B} (\chi') \cdot \widehat{C} (\chi+\chi')}&\le \sum_{\chi, \chi'\in \widehat{\cV}} \abs{\widehat{A}(\chi) \cdot \widehat{A} (\chi') \cdot  \widehat{B} (\chi) \cdot \widehat{B} (\chi')}\cdot  {\delta_1}\\
&= \pbra{ \sum_{\chi\in \widehat{\cV}} \abs{\widehat{A} (\chi) \cdot \widehat{B} (\chi)}}^2\cdot  {\delta_1}\\
&\le  {\delta_1}.
\end{align*}
The last inequality follows from Cauchy-Schwartz and Parseval as in the previous bound. This completes the proof of \cref{lem:fourierlemmaedge}.
\end{proof}

\section{Main Proof}
	 
We use the following Parallel Repetition Criterion which is similar to, but weaker than the one from~\cite{HR20} for the $\GHZ$ game and has a slightly simpler proof. 

Let $\cG$ refer to the $n$-fold parallel repetition of the $\GHZ$ game. Let $\cP=\cQ^n$.

\begin{lemma}[Parallel Repetition Criterion]\label{lem:parallel_repetition_criteria}  Let $c\in(0,1]$ be a constant and $\rho(n):\bbN\to \bbR$ be a function such that $\rho(n)\ge \exp(-n)$. Suppose for all large $n\in \bbN$ and all subsets $E_1, E_2, E_3\subseteq \bbF_2^n$ such that $\cP(E)\ge \rho(n)$ where  $E=E_1\times E_2\times E_3$, we have $\E_{i\sim [n]}\sbra{\val^{(i)}(\cG|E)}\le 1-c$. Then, 
\[\val(\cG)\le \rho(n)^{\Omega(1)}. \]
\end{lemma}

This lemma is proved in~\cite{HR20} under the weaker assumption that there is {\it some} coordinate $i\in[n]$ for which $\val^{(i)}(\cG|E)\le 1-c$. The proof is slightly simpler under our stronger assumption that $\E_{i\sim [n]}\sbra{\val^{(i)}(\cG|E)}\le 1-c$. We prove this in \Cref{proof:parallel_repetition_criteria}. 

Given this criterion, our goal of showing an inverse polynomial bound for $\val(\cG)$ reduces to showing the following. Let $E=E_1\times E_2\times E_3$ be any event such that $\cP(E)= \alpha\ge \frac{1}{n^{1/100}}$ and $n$ be large enough. It suffices to show that $\E_{i\sim [n]}\sbra{\val^{(i)}(\cG|E)} \le 0.95.$ We do this as follows.
	 
Let $\delta=  \tfrac{\alpha^{20}}{n^{1/40}}$. \Cref{cor:pseudorandom} implies the existence of a partition $\Pi$ of $(\bbF_2^n)^3$ into affine subspaces of codimension at most $O\pbra{ \tfrac{1}{\delta^3}}= o(n)$ such that: 
\begin{itemize}
\item Every $\pi\in \Pi$ is of the form $a+\cV^3$ where $\cV\subseteq\bbF_2^n$ is a subspace and $a\in(\F_2^n)^3$.
    \item With probability at least $1-\frac{ \delta}{\cP(E)} \ge 1 - o(1)$ over $\pi\sim \Pi(\cP|E)$, we have $\abs{\widehat{E_i|_{\pi_i}}(\chi)}
\le \delta$ for all $i\in [3]$ and non-zero $ \chi\in \widehat{\cV}$, where $\cV$ is the subspace of $\F_2^n$ for which $\pi$ is an affine shift of $\cV^3$.  
\end{itemize}

 Under the distribution $\Pi(\cP|E)$, the probability that $\pi$ is sampled equals  $\tfrac{(\cP|\pi)(E)\cdot \cP(\pi)}{\cP(E)}$ by Bayes' rule. This implies that the probability that $\pi\sim \Pi(\cP|E)$ satisfies $(\cP|\pi)(E)\le \cP(E)/10$ is at most $1/10$. 
 We will focus on $\pi=\pi_1\times \pi_2\times \pi_3$ that satisfy both these properties, namely, the measure of $E$ under $\cP|\pi$ is significant, furthermore, for all $i\in[3]$, all non-zero Fourier coefficients of the sets $E_i$ restricted to $\pi_i$ are small.
 
 \begin{definition}
     We say that $\pi$ is \textdef{good} if
    \begin{equation}\label{eq:temp3} (\cP|\pi)(E)\ge \alpha/10, \text{ and for all non-zero } \chi\in \widehat{\cV}\text{ and } i\in[3], \text{ we have } \abs{\widehat{E_i|_{\pi_i}}(\chi)}
 \le \delta.
    \end{equation}
 \end{definition}
 By a union bound, a random $\pi \sim \Pi(\cP | E)$ will be good with probability at least $1-\frac{1}{10} - \frac{\delta}{\alpha}$. 
 Fix any such good $\pi = \pi_1 \times \pi_2 \times \pi_3 \in \Pi$, and let $\cV$ be the subspace such that $\pi$ is an affine shift of $\cV^3$.  %Note that $\pi\cap \supp(\cP)$ is non-empty since $\pi$ is sampled from $\Pi(\cP|E)$. Therefore, we may choose $a\in \supp(\cP)$ so that $\pi=a+\cV^3$. Also note that $|\pi_1|=|\pi_2|=|\pi_3|=|\cV|.$ 

For all $z\in E_3\cap \pi_3$, define a (partial) matching $M_{z}$ between $\pi_1$ and $\pi_2$ as follows. For $x\in\pi_1\cap E_1,y\in \pi_2\cap E_2,z\in \pi_3\cap E_3$ such that $x+y=z$, put an edge $(x,y)$. Let $L_{z}$ (resp. $R_{z}$) be the left (resp. right) endpoints of $M_{z}$. Let ${G}=\cup_{z\in E_3\cap \pi_3} M_{z}$ be the bipartite graph between $\pi_1$ and $\pi_2$ obtained by combining edges from the matchings for $z\in E_3\cap \pi_3$. Let $E({G})$ denote the set of edges in ${G}$.  For every edge $e\in E(G)$, we can identify $e$ with a valid input to the $n$-fold $\GHZ$ game that is contained in $E\cap \pi$. Namely, we associate $(x_0,y_0)\in E(G)$ to the input $(x_0,y_0,x_0+y_0)\in \supp(\cP)\cap E\cap \pi$. This is a bijective correspondence because of the way we defined the graph $G$. Under this correspondence, the uniform distribution over edges of $G$ corresponds to the distribution $\cP|E,\pi$. We now introduce the important notion of a bow tie. 
 
 \begin{definition}[Bow Tie]
 We say that a subset of edges $b\subseteq E(G)$ is a \textdef{bow tie} if $b=\{x_0,x_1\}\times \{y_0,y_1\}$ for some $x_0\neq x_1\in \pi_1$, $y_0\neq y_1\in \pi_2$ such that $x_0+y_0=x_1+y_1$ (or equivalently $x_0+y_1=x_1+y_0$). Alternatively, for $z_0=x_0+y_0$ and $z_1=x_0+y_1$, we have $(x_i,y_j,z_k)\in \supp(\cP)$ for all $(i,j,k)\in \supp(\cQ)$.
%  \knote{changed from $(x_i,y_j,z_k)\in \supp(\cP)$ for all $(i,j,k)\in \supp(\cQ)$}\unote{Hm, why? The earlier statement is correct right?} \knote{Yeah sorry corrected. But $i,j,k$ should be 0/1's right?}\unote{They should only be in the support of $\GHZ$ right? }
 
 Let $b=\{x_0,x_1\}\times \{y_0,y_1\}$  be a bow tie.
%  \knote{What is the difference between a bow tie and edges of a bow tie?}\unote{The bowtie is the {\it set} of edges. We probably don't need to define edges OF the bow tie separately, do we? Doesn't hurt though.}
%  \knote{Why isn't $b = \{(x_i,y_j) \,| \, i,j\in\{0,1\}\}$? If it is, what is the difference between edges of $b$ and $b$ itself?} \knote{Ohh do you mean edges of $b$ to be the elements in the set. If so, I don't think we need to define it explicitly.}\unote{Okay.}
As before, we identify $b$ with the indicator vector $b\in \{0,1\}^{E(G)}$ of the edges of $b$, that is,
$b(e)=1$ iff $e\in\{(x_i,y_j):i,j\in\{0,1\}\}$. We use $\tilde{b}$ to denote the uniform distribution on the edges of the bow tie, when viewed as inputs to the $n$-fold $\GHZ$ game. More precisely, $\tilde{b}$ denotes the uniform distribution on $\{(x_i,y_j,x_i+y_j)\,|\, i,j\in\{0,1\}\}$.

We say that $b$ \textdef{differs in the $i$-th coordinate} for $i\in[n]$ if $x_0(i)\neq x_1(i)$, or equivalently, $y_0(i)\neq y_1(i)$, or equivalently, $z_0(i)\neq z_1(i)$. %We say that the \textdef{weight} of $b$ is $|\{ i\in [n]\, |\, x_0(i)\neq x_1(i)\}|$.
 
 \end{definition}  

Let $b$ be a bow tie and $I\subseteq [n]$ be the coordinates on which $b$ differs. The following claim shows that $\val^{(i)}(\cG|\tilde{b})\le 3/4$ for all $i\in I$. The proof is deferred to \Cref{proof:embed}

\begin{claim}\label{claim:embed}
Let $b=\{x_0,x_1\}\times \{y_0,y_1\}$ be a bow tie. Let $I\subseteq [n]$ be the subset of coordinates on which $b$ differs. Then,   $\val^{(i)}(\cG|\tilde{b})\le 3/4$ for all $i\in I$.
\end{claim}

Let $B$ denote the set of all bow ties. Consider the distribution on edges defined by first sampling a uniformly random bow tie from $B$, and then a uniformly random edge from the bow tie. We now provide an alternate description of this distribution. For each $z\in E_3 \cap\pi_3$, define $1_{z}\in \{0,1\}^{|E(G)|}$ as follows. For each $e=(x,y)\in E(G)$, define $1_{z}(e)=1$ if $x$ and $y$ are both matched in $M_{z}$ but not to each other, and define $1_{z}(e)=0$ otherwise. Alternatively, $1_z$ is the indicator of the set $((L_z\times R_z)\setminus M_z) \cap E(G)$. Let $v:= \E_{z\sim E_3\cap \pi_3}[1_z]$. Note that $v$ has $|E(G)|$ coordinates, each of which have non-negative values, so $v$ induces a distribution on $E(G)$. Consider this distribution $\tilde{v}=\frac{v}{\|v\|_1}$ on $E(G)$ defined by normalizing $v$. We show that this distribution is an alternate description of the aforementioned distribution.
% \knote{I guess good to say somewhere here also that $1_z$ is indicator of the set $((L_z\times R_z)\setminus M_z) \cap E(G)$}\unote{Okay, I am adding this.}

\begin{claim} \label{claim:equivalence}
$v=|E_3\cap \pi_3|^{-1}\cdot\pbra{ \sum_{b\in B} b}$. In particular, we can think of the distribution $\tilde{v}:=\frac{v}{\|v\|_1}$ on $E(G)$ as obtained by sampling a uniformly random bow tie $b$ in $G$ and outputting a uniformly random edge of $b$.
\end{claim}

 The proof of this is deferred to \Cref{proof:equivalence}. Our goal now is to show that the distribution $\tilde{v}$ is close to the uniform distribution over edges of $G$. To do so, we study some properties of $G$. Observe that $|E(G)|\triangleq |\cV|^2\cdot \E_{\substack{x\sim \pi_1\\z\sim \pi_3}}\sbra{E_1(x)\cdot E_2(x+z)\cdot E_3(z)}$. We apply \Cref{lem:fourierlemmaedge} with parameters $A=E_1\cap\pi_1,B=E_2\cap\pi_2,C=E_3\cap \pi_3$. Since $\pi\in \supp(\Pi(\cP|E))$, the set $\pi\cap \supp(\cP)$ is non-empty, therefore, we may choose $a\in \supp(\cP)$ so that $\pi=a+\cV^3$. This, along with \Cref{eq:temp3} implies that the first hypothesis of \Cref{lem:fourierlemmaedge} is satisfied. \Cref{lem:fourierlemmaedge} implies that 
	 \begin{align}\label{eq:edgecount}
	\Big | |E(G)|-|\cV|^2\cdot \mu_{\pi_1}(E_1)\cdot\mu_{\pi_2}(E_2)\cdot \mu_{\pi_3}(E_3) \Big | \le |\cV|^2\cdot \delta.
	 \end{align}	 
 We make use of the following bounds on the $\ell_1$ and $\ell_2$ norms of $v$. The proofs of these are by Fourier analysis and are deferred to \Cref{proof:l1norm,proof:l2norm}.

\begin{claim}\label{claim:l1norm} 
	\begin{align}
	\label{eq:l1norm}
	\notag \|v\|_1 &\ge  |\cV|^2 \cdot \pbra{\mu_{\pi_1}(E_1)^2\cdot \mu_{\pi_2}(E_2)^2\cdot  \mu_{\pi_3}(E_3) -3\cdot \delta } \\ &-|\cV|\cdot \pbra{\mu_{\pi_1}(E_1)\cdot \mu_{\pi_2}(E_2)  + 2\cdot\delta\cdot \mu_{\pi_3}(E_3)^{-1}}
	\end{align}
\end{claim}	 	 

\begin{claim}\label{claim:l2norm}
\begin{equation}
\label{eq:l2norm}
\| v\|_2^2\le |\cV|^2\cdot\pbra{ \mu_{\pi_1}(E_1)^3\cdot \mu_{\pi_2}(E_2)^3\cdot \mu_{\pi_3}(E_3) +  10\cdot \sqrt{\delta} }
\end{equation}
\end{claim}
	 
We now bound $\|\tilde{v}\|_2 = \frac{\|v\|_2}{\|v\|_1}$ by plugging in appropriate bounds on $\delta$ and dividing \cref{eq:l2norm} by \cref{eq:l1norm}.   Our choice of $\delta= \alpha^{20}/n^{1/40}$, and our assumption that $\alpha/10\le (\cP|\pi)(E)$ (which in turn is at most $\min_{i\in [3]}\pbra{\mu_{\pi_i}(E_i)}$) implies that $\delta$ is much smaller than any $\mu_{\pi_i}(E_i)$.  In particular, we highlight that
\begin{align*}
\sqrt{\delta}& = o \big (\mu_{\pi_1}(E_1)^3\cdot \mu_{\pi_2}(E_2)^3\cdot \mu_{\pi_3}(E_3) \big )\\
\delta & = o \big (\mu_{\pi_1}(E_1)^2\cdot \mu_{\pi_2}(E_2)^2\cdot \mu_{\pi_3}(E_3) \big ) \\
\delta &= o \big (\mu_{\pi_1}(E_1) \cdot \mu_{\pi_2}(E_2) \cdot \mu_{\pi_3}(E_3) \big ) 
\end{align*}

Furthermore, since $|\cV|= 2^{\Omega(n)}$ and $1 \ge \mu_{\pi_i}(E_i) = \Omega(\alpha) = n^{-O(1)}$, we have 
% \[
% |\cV|^2\cdot \mu_{\pi_1}(E_1)^2\cdot \mu_{\pi_2}(E_2)^2\cdot \mu_{\pi_3}(E_3) = \omega \big ( |\cV|\cdot \mu_{\pi_1}(E_1)\cdot \mu_{\pi_2}(E_2) \big ).
% \]
\[
|\cV|\cdot \mu_{\pi_1}(E_1)\cdot \mu_{\pi_2}(E_2) = o\left(|\cV|^2\cdot \mu_{\pi_1}(E_1)^2\cdot \mu_{\pi_2}(E_2)^2\cdot \mu_{\pi_3}(E_3)\right).
\]

Thus the dominant term on the right-hand side of \cref{eq:l1norm} is $|\cV|^2 \cdot \mu_{\pi_1}(E_1)^2\cdot \mu_{\pi_2}(E_2)^2\cdot  \mu_{\pi_3}(E_3)$, and the dominant term on the right-hand side of \cref{eq:l2norm} is $|\cV|^2\cdot \mu_{\pi_1}(E_1)^3\cdot \mu_{\pi_2}(E_2)^3\cdot \mu_{\pi_3}(E_3)$.  More precisely, we have 
\begin{align}\label{lem:l1l2norm}
    \|v\|_1&\ge  (1 - o(1)) \cdot |\cV|^2\cdot \mu_{\pi_1}(E_1)^2\cdot\mu_{\pi_2}(E_2)^2\cdot \mu_{\pi_3}(E_3) \\
 \|v\|_2^2&\le  (1 + o(1)) \cdot |\cV|^2\cdot \mu_{\pi_1}(E_1)^3\cdot\mu_{\pi_2}(E_2)^3\cdot \mu_{\pi_3}(E_3).
\end{align}  
This implies that \begin{equation}
\label{eq:normalized-l2-bound}
    \|\tilde{v}\|_2^2 = \frac{\|v\|_2^2}{\|v\|_1^2}\le \frac{1 + o(1)}{|\cV|^2 \cdot \mu_{\pi_1}(E_1) \cdot \mu_{\pi_2}(E_2) \cdot \mu_{\pi_3}(E_3)}
\end{equation}

In comparison, \cref{eq:edgecount} gave that 
\[
|E(G)| \in (1 \pm o(1)) \cdot |\cV|^2\cdot \mu_{\pi_1}(E_1)\cdot \mu_{\pi_2}(E_2)\cdot \mu_{\pi_3}(E_3). 
\]
Thus we can rewrite \cref{eq:normalized-l2-bound} as
\begin{equation}
\label{eq:relative-normalized-l2-bound}
\|\tilde{v}\|_2 \leq \frac{1 + o(1)}{\sqrt{|E(G)|}}
\end{equation}
%\knote{Changed slightly}

This, together with the fact that by construction $\|\tilde{v}\|_1 = 1$, is sufficient to deduce that $\tilde{v}$ is close to the ``uniform distribution" vector $\tilde{u} \eqdef (\frac{1}{|E(G)|}, \ldots, \frac{1}{|E(G)|})$.  More formally, we have:
\begin{fact}
\label{fact:uniformity-criterion}
Suppose that $\tilde{v}\in \bbR^m$ is an $m$-dimensional vector such that $\|\tilde{v}\|_1 = 1$, and $\|\tilde{v}\|_2 = \frac{1 + \beta}{\sqrt{m}}$ for some $\beta \in [0,1]$. Then
\[
\|\tilde{v}-\tilde{u}\|_1 \le \sqrt{3 \beta},
\]
where $\tilde{u}$ denotes the vector $(\frac{1}{m}, \ldots, \frac{1}{m})$. %\knote{Changed slightly}
\end{fact}
\noindent
The proof of \cref{fact:uniformity-criterion} is deferred to \cref{sec:uniformity-criterion}

Applying \cref{fact:uniformity-criterion} to \cref{eq:relative-normalized-l2-bound} shows that $\dtv(\tilde{v},\tilde{u}) = o(1)$. In other words, a uniformly random edge of a uniformly random bow tie is distributed close to uniformly on $E(G)$. 

We now show that a typical bow tie differs in a considerable fraction of coordinates.
\begin{claim}
\label{claim:differsalot} $\Pr_{\substack{i\sim [n]\\ b\sim B}}[b\text{ differs in }i\text{-th coordinate}]\ge 1/3 - o(1)$.
\end{claim}
\noindent
The proof of \cref{claim:differsalot} is deferred to \Cref{proof:differsalot}.
 
% \Cref{claim:embed}, along with \Cref{claim:differsalot} implies that
% $\Pr_{\substack{i\sim [n]\\ b\sim B}}[\val^{(i)}_{\GHZ}(\tilde{b})\le 3/4]\ge 0.3.$ For those $i\in[n]$ and $b\in B$ such that $b$ doesn't differ at the $i$-th coordinate, we bound $\val^{(i)}_{\GHZ}(\tilde{b})$ by 1. This, along with \Cref{claim:equivalence} implies that $\E_{i\in [n]}\sbra{\val^{(i)}_{\GHZ}(\tilde{v})}\le \E_{\substack{i\sim [n]\\ b\sim B}}[\val^{(i)}_{\GHZ}(\tilde{b})]\le 0.75\times 0.3+1\times 0.7\le 0.925.$ Since $\dtv(\tilde{u},\tilde{v})\le 0.071$ and $\tilde{u}$ corresponds to $\cP|\pi,E$, this implies that $\E_{i\sim [n]}\sbra{\val^{(i)}(\cP|\pi,E )}\le 0.925+1\times 0.071\le 0.996$. Since this holds with probability at least $1-\frac{3\cdot \delta^2}{\cP(E)}-1/1000\ge 0.998$ over $\pi\sim \Pi(\cP|E)$ (for large $n\in\bbN$), we have $\E_{i\sim [n]}\sbra{\val^{(i)}(\cP|E )}\le \underset{\substack{i\sim [n]\\\pi\sim \Pi(\cP|E)}}{\E}\sbra{\val^{(i)}(\cP|E,\pi)}\le 0.996+1\times 0.002< 0.999.$ This completes the proof.

\Cref{claim:embed}, along with \Cref{claim:differsalot} implies that
$\Pr_{\substack{i\sim [n]\\ b\sim B}}[\val^{(i)}(\cG|\tilde{b})\le 3/4]\ge 1/3-o(1)\ge 0.3.$ For those $i\in[n]$ and $b\in B$ such that $b$ doesn't differ at the $i$-th coordinate, we bound $\val^{(i)}(\cG|\tilde{b})$ by 1. This, along with \Cref{claim:equivalence} implies that $\E_{i\sim [n]}\sbra{\val^{(i)}(\cG|\tilde{v})}\le \E_{\substack{i\sim [n]\\ b\sim B}}[\val^{(i)}(\cG|\tilde{b})]\le 0.75\times 0.3+1\times 0.7\le 0.925.$ Since $\dtv(\tilde{u},\tilde{v})\le o(1)$ and $\tilde{u}$ corresponds to $\cP|\pi,E$, this implies that $\E_{i\sim [n]}\sbra{\val^{(i)}(\cG|\pi,E )}= 0.925+o(1) \leq 0.93$. Since $\pi \sim \Pi(\cP | E)$ is good with probability at least $1- \delta\cdot\alpha^{-1}-1/10\ge 0.9-o(1) \geq 0.8$, we have $\E_{i\sim [n]}\sbra{\val^{(i)}(\cG|E )}\le \underset{\substack{i\sim [n]\\\pi\sim \Pi(\cP|E)}}{\E}\sbra{\val^{(i)}(\cG|E,\pi)}
% \le \left(0.9- o(1)\right)\times (0.925+o(1)) + (0.1+o(1))\times 1 < 0.94+o(1),$
\le 0.8\times 0.93 + 0.2 \times 1 < 0.95$. This, along with \Cref{lem:parallel_repetition_criteria} completes the proof.
% \knote{made minor changes to values}
%\jnote{the notation with all the $\val$ variants should be made consistent}

\subsection[Proof of Lemma 5.1]{Proof of \Cref{lem:parallel_repetition_criteria}}\label{proof:parallel_repetition_criteria}

%\jnote{heads up that I'm skipping the proofreading of this subsection for now, because the lemma follows from HR20}
\begin{proof}[Proof of \cref{lem:parallel_repetition_criteria}]
	
Let $\cP=\cQ^n$. Choose the largest integer $m\ge 0$ such that $32^{-m}\ge \rho(n)\cdot \tfrac{2}{c}$. Note that $m=\Theta(\log(1/\rho(n)))$. Fix any deterministic product strategy $\bar{f}=(\bar{f}_1,\bar{f}_2,\bar{f}_3)$ for the players where $\bar{f}_i:\F_2^n\to \F_2^n$ denotes the strategy for the $i$-th player. Let $Y_i=\bar{f}_i(X_i)\in \F_2^n$ denote the output of player $i$ on input $X_i$. Let $\{j_1, \ldots, j_m\} \subseteq [n]$ be a set of coordinates. Let $W_i$ denote the event of winning the $\GHZ$ game in the $j_i$-th coordinate under the strategy $\bar{f}$ and let $W_{\le i}:=W_1\wedge \ldots\wedge W_{i}$. Observe that
\[ \val(\cG,\bar{f})\le \prod_{i=0}^{m-1} \Pr\sbra{ W_{i+1}\mid W_{\le i}}.\]
 
We show how to construct a sequence of coordinates so that every term in the above product is at most $1-c/2$. This would imply that  $\val(\cG)\le  (1-c/2)^{\Theta(\log(1/\rho(n))}=\rho(n)^{\Omega(1)}.$  Fix any $i\in\{0,\ldots,m-1\}$ and assume that we have found $j_1,\ldots,j_i$. Let $X\sim \cP$ and $X_{\le i}$ denote $X$ restricted to the coordinates $\{j_1,\ldots,j_{i}\}$. Let $Y_{\le i}$ denote the outputs of the players restricted to the coordinates $\{j_1,\ldots,j_{i}\}$. Let $Z_{\le i}=(X_{\le i},Y_{\le i})$. Since $W_{\le i}$ is a function of $Z_{\le i}$, we have
\begin{equation} \Pr\sbra{ W_{i+1}\mid W_{\le i} }= \underset{z_{\le  i}\sim Z_{\le i}|W_{\le i} }{\E}\sbra{\Pr\sbra{W_{i+1}\mid Z_{\le i}=z_{\le i}}}\le \underset{z_{\le  i}\sim Z_{\le i}|W_{\le i} }{\E}  \sbra{\val^{(j_{i+1})}\pbra{\cG| Z_{\le i}=z_{\le i} }}. \label{eq:parallelrepetition} \end{equation}
Let $F=F(z_{\le i})$ denote the event that  $\cP\sbra{Z_{\le i}=z_{\le i}|W_{\le i}}\ge \tfrac{c}{2}\cdot  \tfrac{1}{N}$ where $N=32^i\ge \supp(Z_{\le i})$. We argue that $F$ occurs with probability at least $1-c/2$. This is because we are sampling $z_{\le i}$ with probability $\cP[Z_{\le i}=z_{\le i}|W_{\le i}]$, hence the measure of $z_{\le i}$ for which $\cP[Z_{\le i}=z_{\le i}|W_{\le i}]\le \tfrac{c}{2}\cdot \tfrac{1}{N}$ is at most $\tfrac{c}{2}$. Fix any $z_{\le i}$ such that $F$ holds. Our choice of $m$ implies that $\tfrac{1}{N}\cdot \tfrac{c}{2}\ge \rho(n)$. Note that we can express the distribution $\cP|Z_{\le i}=z_{\le i}$ as $\cP|E$ where $E=E_1\times E_2\times E_3$ for $E_1,E_2,E_3\subseteq \bbF_2^n$ and $\cP(E)\ge \rho(n)$. The hypothesis of \Cref{lem:parallel_repetition_criteria} implies that $\E_{j\sim [n]}\sbra{\val^{(j)}\pbra{\cG|Z_{\le i}=z_{\le i}}}\le 1-c$. This implies that
\begin{align*}
\underset{\substack{z_{\le  i}\sim Z_{\le i}|W_{\le i} \\j\sim [n]}}{\E}  \sbra{\val^{(j)}\pbra{\cG| Z_{\le i}=z_{\le i} }}&\le \Pr_{z_{\le  i}\sim Z_{\le i}|W_{\le i} }[\neg F] +\underset{\substack{z_{\le  i}\sim Z_{\le i}|W_{\le i},F\\j\sim [n]} }{\E}  \sbra{\val^{(j)}\pbra{\cG| Z_{\le i}=z_{\le i}}}\\
&\le \tfrac{c}{2}+1-c = 1-\tfrac{c}{2}.
\end{align*}
By linearity of expectation, we can fix a $j\in[n]$ such that $\underset{z_{\le  i}\sim Z_{\le i}|W_{\le i}}{\E}  \sbra{\val^{(j)}\pbra{\cG| Z_{\le i}=z_{\le i} }}\le 1-\tfrac{c}{2}$. Note that $j\notin \{j_1,\ldots,j_{i}\}$ since we already win the game on these coordinates. This, along with \cref{eq:parallelrepetition} completes the proof.
\end{proof}

\subsection[Proof of Claim 5.2]{Proof of \Cref{claim:embed}}\label{proof:embed}

\begin{proof}[Proof of \cref{claim:embed}]
% \knote{Can we say that by swapping $x_0, x_1$ and $y_0, y_1$, we can assume $x_0(i)=y_0(i)=0$. Then $\phi_1(a)=x_a$, and so on, and don't need to use $\phi$'s basically.}\unote{That makes sense.}
Let $i\in I$. Since the bow tie $b$ differs in the $i$-th coordinate, we have \[\{x_0(i),x_1(i)\}=\{y_0(i),y_1(i)\} =\{z_0(i),z_1(i)\}=\{0,1\}.\] We may thus assume without loss of generality that $x_0(i)=y_0(i)=0$. Define embeddings $\phi_1:\F_2\to \{x_0,x_1\}$, $\phi_2:\F_2\to \{y_0,y_1\}$ and $\phi_3:\F_2\to \{z_0,z_1\}$ at $a\in \F_2$ by $\phi_1(a)=x_a$, $\phi_2(a)=y_a$ and $\phi_3(a)=z_a$.
It follows for all $a\in \{0,1\}$ and $j\in[3]$, we have $(\phi_j(a))(i)=a$. In particular, for $\phi=\phi_1\times \phi_2\times \phi_3$, the distribution $\phi(\cQ)$ is exactly the distribution $\tilde{b}$. 
Given any strategies $\bar{f}_1,\bar{f}_2,\bar{f}_3:\F_2^n\to \F_2^n$ for the players for the $n$-fold $\GHZ$ game restricted to the query distribution $\tilde{b}$, the functions $\phi_1,\phi_2,\phi_3$ induce a strategy for the $\GHZ$ game as follows. Define $f_j:\F_2\to\F_2$ by $f_j(a)=(\bar{f}_j(\phi_j(a)))(i).$ The success probability of the strategy $f_1\times f_2\times f_3$ on the distribution $\cQ$ is exactly the success probability in the $i$-th coordinate of the strategy $\bar{f}_1\times \bar{f}_2\times \bar{f}_2$ on the distribution $\tilde{b}$. It follows that $\val^{(i)}(\cG|\tilde{b})\le 3/4$.

\end{proof}

\subsection[Proof of Claim 5.3]{Proof of \Cref{claim:equivalence}}\label{proof:equivalence}

\begin{proof}[Proof of \cref{claim:equivalence}]
	Fix any $e\in E(G),e=(x_0,y_0)$. This implies that $x_0\in E_1\cap \pi_1,y_0\in E_2\cap \pi_2$  and $z_0:=x_0+y_0 \in E_3\cap \pi_3$. Note that $v(e)=\Pr_{z\sim E_3\cap\pi_3}\sbra{(x_0,y_0)\in (L_z\times R_z)\setminus M_z}$. For any $z_1\in E_3\cap \pi_3$, 
	\begin{align*} e\in (L_{z_1}\times R_{z_1})\setminus M_{z_1}&\iff x_0+z_1\in E_2\cap \pi_2, y_0+z_1\in E_1\cap\pi_1,z_1\neq z_0 \\
		&\iff x_0,x_1\in E_1\cap\pi_1, y_0,y_1\in E_2\cap\pi_2,z_1\neq z_0\in E_3\cap \pi_3\\
		&\text{ where }x_1:=y_0+z_1,y_1:=x_0+z_1\\
		&\iff \{x_0,x_1\}\times\{y_0,y_1\}\text{ is a bow tie}\\
		&\text{ where }x_1:=y_0+z_1,y_1:=x_0+z_1.\end{align*}
	This implies that for all $e=(x_0,y_0)\in E(G)$ and $z_1\in E_3\cap \pi_3$, we have $1_{z_1}(e)=1$ if and only if  $b=\{x_0,x_1\}\times\{y_0,y_1\}$ is a bow tie. Observe that as we vary $z_1\in E_3\cap \pi_3$, we obtain all possible bow ties that contain the edge $e$, i.e. the bow ties $b$ for which $b(e)\neq 0$. This implies that $v\triangleq \E_{z_1\sim E_3\cap\pi_3}\sbra{1_z}= |E_3\cap \pi_3|^{-1}\cdot\pbra{ \sum_{b\in B} b}.$
\end{proof}

\subsection[Proof of Claim 5.4]{Proof of \Cref{claim:l1norm}}\label{proof:l1norm}

For ease of notation, we define weight functions as follows. 
\begin{definition}[Weight functions] \label{def:wtfunctions} Let $\cP=\cQ^n$. For $z\in \pi_3$, let
	\[\wt_\pi(z):=\Pr_{X\sim \cP}\sbra{(X_1\in E_1\text{ and }X_2\in E_2)|(X\in\pi\text{ and } X_{3}=z)}=\underset{x\sim\pi_1}{\E}\sbra{E_1(x)E_2 (x+z)}. \]
\end{definition}

\begin{proof}[Proof of \Cref{claim:l1norm}]

Let $z\in E_3\cap\pi_3$. Note that $\wt_{\pi}(z)=\mu_{\pi_1}(L_z)=\mu_{\pi_2}(R_z)$. Observe that $\|1_z\|_1=|E(G)\cap (L_z\times R_z)\setminus M_z|$. We apply \Cref{lem:fourierlemmaedge} with parameters $A=L_z\cap\pi_1,B=R_z\cap\pi_2,C=E_3\cap \pi_3$. The first hypothesis of \Cref{lem:fourierlemmaedge} is satisfied due to \cref{eq:temp3}. \Cref{lem:fourierlemmaedge} implies that
\begin{align*} \abs{E(G)\cap (L_{z}\times R_{z})} &\triangleq |\cV|^2\cdot  \underset{\substack{z'\sim \pi_3\\x\sim \pi_1}}{\E}\sbra{L_z(x)\cdot R_z(x+z')\cdot E_3(z')}\\
&\ge |\cV|^2\cdot \pbra{ \mu_{\pi_1}(L_z)\cdot \mu_{\pi_2}(R_z)\cdot \mu_{\pi_3}(E_3)-\delta }\\
&\triangleq  |\cV|^2 \cdot\pbra{\wt_{\pi}(z)^2\cdot \mu_{\pi_3}(E_3) - \delta }.
\end{align*}
Similarly, $ |M_z|\triangleq |\cV|\cdot \E_{x\sim \pi_1}\sbra{E_1(x)\cdot E_2(x+z)}= |\cV|\cdot \wt_\pi(z)$. We apply
 \Cref{lem:fourierlemmaedge} with parameters $A=E_1,B=E_2,C=E_3$. All the hypothesis are satisfied due to \cref{eq:temp3}.  \Cref{lem:fourierlemmaedge}, along with conditioning $z\sim \pi_3$ on $z\in E_3$ implies that
\begin{equation} \label{eq:temp4} \abs{\E_{z\sim E_3\cap\pi_3}\sbra{\wt_\pi(z)^2}-\mu_{\pi_1}(E_1)^2\cdot \mu_{\pi_2}(E_2)^2}\le 2\cdot\delta \cdot \mu_{\pi_3}(E_3)^{-1}.
\end{equation}
% \knote{the expectations should be inside the     absolute value and only on $\wt_\pi(z)$ terms? Similarly in next claim.}\unote{Yes}
\[\label{eq:temp5}\abs{ \E_{z\sim E_3\cap \pi_3}\sbra{\wt_\pi(z)}-\mu_{\pi_1}(E_1)\cdot \mu_{\pi_2}(E_2)}\le 2\cdot \delta \cdot \mu_{\pi_3}(E_3)^{-1}.\]
Substituting this in the previous inequalities and taking an expectation over $z\sim E_3\cap\pi_3$,
\begin{align*}\|v\|_1= \E_{z\sim E_3\cap \pi_3}\sbra{\|1_z\|_1} 
	&=\E_{z\sim E_3\cap \pi_3}\sbra{\abs{E(G)\cap (L_z\times R_z)}-\abs{M_z}}\\
	&\ge |\cV|^2 \cdot\pbra{\E_{z\sim E_3\cap\pi_3}  \sbra{\wt_{\pi}(z)^2}\cdot \mu_{\pi_3}(E_3) - \delta }  - |\cV|\cdot \E_{z\sim E_3\cap \pi_3}\sbra{\wt_\pi(z)}\\
	&\ge  |\cV|^2 \cdot \pbra{\mu_{\pi_1}(E_1)^2\cdot \mu_{\pi_2}(E_2)^2\cdot  \mu_{\pi_3}(E_3) -3\cdot\delta } \\
	&-|\cV|\cdot \pbra{\mu_{\pi_1}(E_1)\cdot \mu_{\pi_2}(E_2)  + 2\cdot\delta\cdot \mu_{\pi_3}(E_3)^{-1}}.  \qedhere
\end{align*} 
\end{proof}

\subsection[Proof of Claim 5.5]{Proof of \Cref{claim:l2norm}}\label{proof:l2norm}

\begin{proof}[Proof of \cref{claim:l2norm}] Define $\wt_\pi(\cdot)$ as in the proof of \Cref{claim:l1norm}. Let $z,z'\in E_3\cap\pi_3$. Observe that $\langle 1_z,1_{z'}\rangle=\abs{E(G)\cap ((L_z\cap L_{z'})\times(R_z\cap R_{z'}))\setminus(M_z\cup M_{z'})}$. We apply \Cref{lem:fourierlemmaedge} with parameters $A=L_z\cap L_{z'}\cap \pi_1$, $B=R_z\cap R_{z'}\cap \pi_2$ and $C=E_3\cap \pi_3$. The first hypothesis is satisfied due to \cref{eq:temp3}. \Cref{lem:fourierlemmaedge} implies that
\begin{align*} \langle 1_z,1_{z'}\rangle&= \big |E(G)\cap ((L_z\cap L_{z'})\times (R_z\cap R_{z'}))\setminus (M_z\cup M_{z'}) \big | \\
	&\le |\cV|^2\cdot \pbra{\mu_{\pi_1}(L_z\cap L_{z'})\cdot \mu_{\pi_2}(R_z\cap R_{z'})\cdot \mu_{\pi_3}(E_3)+\delta  }.\end{align*}
Taking an expectation over $z'\sim E_3\cap \pi_3$ and applying Cauchy-Schwartz  yields that
\begin{align*} &\E_{z'\sim E_3\cap\pi_3}\sbra{\langle 1_z,1_{z'}\rangle}\\
 	&\le |\cV|^2\cdot  \E_{z'\sim E_3\cap\pi_3}\sbra{\mu_{\pi_1}(L_z\cap L_{z'})\cdot \mu_{\pi_2}(R_z\cap R_{z'})\cdot \mu_{\pi_3}(E_3)+ \delta }\\
 	&\le |\cV|^2\cdot \pbra{ \sqrt{\E_{z'\sim E_3\cap \pi_3}\sbra{\mu_{\pi_1}(L_z\cap L_{z'})^2}}\cdot \sqrt{\E_{z'\sim E_3\cap \pi_3}\sbra{\mu_{\pi_2}(R_z\cap R_{z'})^2}}\cdot \mu_{\pi_3}(E_3)+ \delta }.
\end{align*}
Observe that $\mu_{\pi_1}(L_z\cap L_{z'})=\E_{x\sim \pi_1}\sbra{L_z(x)E_2(x+z')}$ for all $z'\in E_3\cap \pi_3$. We now apply \Cref{lem:fourierlemmaedge} with parameters $A=L_z\cap \pi_1,B=E_2\cap \pi_2,C=E_3\cap \pi_3$. All the hypotheses are satisfied due to \cref{eq:temp3}. \Cref{lem:fourierlemmaedge}, along with the aforementioned observation implies that
\[ \abs{\E_{z'\sim E_3\cap \pi_3}\sbra{ \mu_{\pi_1}(L_z\cap L_{z'})^2}- \mu_{\pi_1}(L_z)^2\cdot \mu_{\pi_2}(E_2)^2   }\le 2\cdot \delta\cdot \mu_{\pi_3}(E_3)^{-1} .\]
An analogous inequality holds for $|R_z\cap R_{z'}|$. Substituting this in the previous inequality and using the fact that $\sqrt{a+b}\le \sqrt{a}+\sqrt{b}$, we have
\begin{align*} &\E_{z'\sim E_3\cap\pi_3}\sbra{\langle 1_z,1_{z'}\rangle}
	\\&\le |\cV|^2\cdot \pbra{ \pbra{ \mu_{\pi_1}(L_z)\cdot \mu_{\pi_2}(E_2) +\sqrt{\tfrac{2\cdot \delta}{\mu_{\pi_3}(E_3)}}} \cdot  \pbra{ \mu_{\pi_2}(R_z)\cdot \mu_{\pi_1}(E_1) +\sqrt{\tfrac{2\cdot \delta}{\mu_{\pi_3}(E_3)}}} \cdot \mu_{\pi_3}(E_3)+ \delta  }\\ 
	&\le|\cV|^2\cdot\pbra{ \mu_{\pi_1}(L_z)\cdot \mu_{\pi_2}(R_z)\cdot \mu_{\pi_1}(E_1)\cdot \mu_{\pi_2}(E_2)\cdot \mu_{\pi_3}(E_3) +  8\cdot \sqrt{\delta} }\\
	&=|\cV|^2\cdot\pbra{ \wt_{\pi}(z)^2\cdot \mu_{\pi_1}(E_1)\cdot \mu_{\pi_2}(E_2)\cdot \mu_{\pi_3}(E_3) +  8\cdot \sqrt{\delta} }.
\end{align*}
We now take an expectation over $z\sim E_3\cap\pi_3$ and use \cref{eq:temp4} to conclude that
\[
\E_{z,z'\sim E_3\cap\pi_3}\sbra{\langle 1_z,1_{z'}\rangle}\le |\cV|^2\cdot\pbra{ \mu_{\pi_1}(E_1)^3\cdot \mu_{\pi_2}(E_2)^3\cdot \mu_{\pi_3}(E_3) + 10\cdot \sqrt{\delta}}.  \qedhere
\]
\end{proof}
 
\subsection[Proof of Fact 5.6]{Proof of \cref{fact:uniformity-criterion}}
\label{sec:uniformity-criterion}
\begin{proof}[Proof of \cref{fact:uniformity-criterion}]
\begin{align*}
\|\tilde{v}-\tilde{u}\|_2^2 & = \langle \tilde{v} - \tilde{u}, \tilde{v} - \tilde{u} \rangle \\ 
& = \|\tilde{v}\|_2^2+\|\tilde{u}\|_2^2- 2 \langle \tilde{u}, \tilde{v} \rangle \\
& = \frac{1 + 2 \beta + \beta^2}{m} + \frac{1}{m} - \frac{2}{m}   \\
& = \frac{2 \beta + \beta^2}{m} \leq \frac{3\beta}{m}.
\end{align*}
Finally, we bound the $\ell_1$ distance in terms of the $\ell_2$ distance:
\[
\|\tilde{v}-\tilde{u}\|_1 \le \|\tilde{v}-\tilde{u}\|_2 \cdot \sqrt{m} \le \sqrt{3\beta}.  \qedhere
\]
\end{proof}

 \subsection[Proof of Claim 5.7]{Proof of \Cref{claim:differsalot}}\label{proof:differsalot}
 
  \begin{proof}[Proof of \cref{claim:differsalot}]
 	 It suffices to show that a random $b\sim B$ differs in less than $n/3$ coordinates with probability at most $2^{-\Omega(n)}=o(1)$.

	The Chernoff bound implies that $\Pr_{x_0, x_1 \sim \F_2^n}\sbra{\hwt(x_0+x_1)< n/3} \leq 2^{-\Omega(n)}$. We condition on $x_0,x_1\in \pi_1$ to conclude that $\Pr_{x_0, x_1 \sim \pi_1}\sbra{\hwt(x_0+x_1)< n/3} \leq 2^{-\Omega(n)}\cdot  \tfrac{2^{2n}}{|\cV|^2}$.
	
	Let $b = \{x_0, x_1\}\times \{y_0, y_1\}$ be a bow tie. By definition, we have $y_1=x_0+x_1+y_0$. In particular, the bow tie $b$ is uniquely identified by $x_0, x_1, y_0$. This implies that the probability that a random $b\sim B$ differs in less than $n/3$ coordinates is precisely
	\begin{align*}
		&\frac{	 |\cV|^3}{|B|} \Pr_{\substack{x_0, x_1 \sim \pi_1\\ y_0 \sim \pi_2\\ y_1=x_0+x_1+y_0}} \sbra{\{x_0, x_1\}\times \{y_0, y_1\}\in B\text{ and } \hwt(x_0+x_1)< n/3}
	\\
		\le\ & 	\frac{ |\cV|^3}{|B|} \Pr_{x_0, x_1 \sim \pi_1} \sbra{\hwt(x_0+x_1)< n/3} 
		\\
		\le\ & 	\frac{ |\cV|^3}{|B|}\cdot  2^{-\Omega(n)} \cdot \frac{2^{2n}}{|\cV|^2}
	\end{align*}
	
		Recall that $v= \E_{z\sim E_3\cap \pi_3}[1_z] = \frac{1}{\mu_{\pi_3}(E_3) \cdot |\cV|} \sum_{z\in E_3\cap \pi_3} 1_z$, where for each $e$, $\sum_{z\in E_3\cap \pi_3} 1_z (e)$ equals the number of bow ties containing the edge $e$.
	Since each bow tie contains 4 edges, we have that $\|v\|_1 = \frac{4}{\mu_{\pi_3}(E_3)\cdot|\cV|}\cdot |B|$.
	Then, equation (\ref{lem:l1l2norm}) implies that
	\[ |B| \ge \frac{1}{8} \cdot |\cV|^3\cdot \mu_{\pi_1}(E_1)^2\cdot\mu_{\pi_2}(E_2)^2\cdot \mu_{\pi_3}(E_3)^2 \geq \frac{1}{8}\cdot |\cV|^3\cdot \alpha^6. \]
	
    This implies that $\tfrac{|\cV|^3}{|B|}\le 8/ \alpha^6$. Recall that $\alpha \geq n^{-O(1)}$ and the co-dimension of $\cV$ is $o(n)$. This implies that $\tfrac{2^{2n}}{|\cV|^2}=2^{o(n)}$.  This along with the above calculation implies that the probability that a uniformly random $b\sim B$ differs in less than $n/3$ coordinates is at most $\frac{8\cdot 2^{-\Omega(n)}}{ \alpha^6}\cdot 2^{o(n)} = 2^{-\Omega(n)}$. This completes the proof.
 \end{proof}

\appendix

\section{Proof of \cref{cor:pseudorandom}}\label{app:pseudorandom}
 
Recall the statement of the Proposition:
\begin{proposition} 
Let $\cP=\cQ^n$. Let $E=E_1\times E_2\times E_3 \subseteq (\F_2^n)^3$ be such that $\cP(E) = \alpha$. For all $\delta>0$, there exists an affine partition $\Pi$ of $(\bbF_2^n)^3$ of codimension at most $\frac{3}{\delta^3}$ such that the following holds. With probability at least $1-\frac{\delta}{\alpha}$ over $\pi\sim \Pi(\cP|E)$, 
\begin{equation}\label{property} \text{ for all }i\in[3]\text{ and non-zero }\chi \in \widehat{\cV},\quad  \abs{\widehat{E_i|_{ \pi_i}}(\chi)}\le   \delta \end{equation}
where $\pi$ is of the form $\pi_1\times\pi_2\times\pi_3$ for affine shifts $\pi_1,\pi_2,\pi_3$ of some subspace $\cV$ of $\F_2^n$.\label{prop:appendix}
\end{proposition}

Recall that $E_i|_{ \pi_i}: \pi_i \to \{0,1\}$ denotes the indicator function of $E_i$ restricted to the subspace $\pi_i\subseteq \F_2^n$. The main idea behind the proof of the above proposition is to keep dividing the space based on Fourier coefficients that violate the required condition. A simple potential function argument shows that this cannot be repeated too many times.

For ease of notation, we introduce the following notation. Let $\cV$ be a subspace of $\F_2^n$, $a_1,a_2,a_3\in \F_2^n$ , $a=(a_1,a_2,a_3)$ and  $\pi=\pi_1\times \pi_2\times \pi_3$ where $\pi_i=a_i+\cV$. Let $\chi\in\widehat{\cV}$. We define $\chi_a^3:\pi\to \{-1,1\}^3$ for all $w_i\in \pi_i$ by
	\[ \chi_a^3 \begin{pmatrix} w_1\\ w_2 \\ w_3 \end{pmatrix}\eqdef \begin{pmatrix} \chi_{a_1}(w_1)\\ \chi_{a_2}(w_2) \\ \chi_{a_3}(w_3) \end{pmatrix}=\begin{pmatrix} \chi(a_1+w_1)\\ \chi(a_2+w_2) \\ \chi(a_3+w_3) \end{pmatrix}.  \]
	
\begin{proof}[Proof of \Cref{prop:appendix}]
	We will show that the desired property, namely \Cref{property} holds with probability at least $1-\delta$ over $\pi \sim \Pi(\cP)$.
	Since $\cP(E)=\alpha$, conditioning on the event $E$ implies that this property holds with probability at least $1-\frac{\delta}{\alpha}$ over $\pi\sim \Pi(\cP|E)$.
	
	We construct the desired partition $\Pi$ iteratively, starting with $\Pi_1=\{(\bbF_2^n)^3\}$.
	At each step $t\geq 1$, if with probability more than $\delta$ over $\pi\sim\Pi_t(\cP)$, \Cref{property} does not hold, then we refine the partition $\Pi_t$ as follows. Consider any $\pi \in \Pi_t$ and let $\pi  = a + \cV^3$ for some subspace  $\cV\subseteq \F_2^n$ and $a\in (\F_2^n)^3$. Choose $i\in [3]$ and $\chi \in \widehat{\cV} \setminus \{\emptyset\}$ such that $\abs{\widehat{E_i|_{ \pi_i}}(\chi)}$ is maximized.
	We partition the subspace $\pi$ into subspaces $\pi_z\eqdef\{x \in \pi: \chi^3_a(x) = z \}$ for each $z\in \{-1,1\}^3$. Let $\Pi_{t+1}=\{ \pi_z: \pi\in \Pi_t, z\in\{-1,1\}^3\}$. This refinement step produces an affine partition of codimension one more than before. We argue that this refinement does not occur too often, by a potential function argument.

	Define a potential function $\Phi=\sum_{i=1}^3 \Phi_i$  where $\Phi_i$ is defined for an affine partition $\Pi$ by 
 	\[ \Phi_i(\Pi)\eqdef\E_{\substack{\pi\sim \Pi(\cP)\\\pi = \pi_1\times \pi_2 \times \pi_3}}\sbra{ \abs{\widehat{E_i|_{ \pi_i}}(\emptyset )}^2 }=\E_{\substack{\pi\sim \Pi(\cP)\\\pi = \pi_1\times \pi_2 \times \pi_3}}\sbra{\mu_{\pi_i}(E_i)^2}. \]  
	By definition, for every $\Pi$, it holds that $0 \leq \Phi(\Pi) \leq 3$.
	We show that in each step of the refinement, $\Phi$ increases by at least $\delta^3$, and hence the process must stop after at most $3/\delta^3$ steps.
	This is proved as follows.
	
	Consider any step $t$ of the refinement. 
	Consider any subspace $\pi = a + \cV^3 \in \supp(\Pi_t(\cP))$ that contributes to the potential function at time $t$. The subspace $\pi$ is partitioned into subspaces $\pi_z:z\in\{-1,1\}^3$ at time $t$. Since $\cP$ is uniform over $\{y\in (\F_2^n)^3: y_1+y_2+y_3 = 0\}$, we have $\pi_z\in \supp(\Pi_{t+1}(\cP))$ if and only if $z_1\cdot z_2\cdot z_3=1$. Hence, to analyze the change in potential function from time $t$ to $t+1$, it suffices to focus on subspaces $\pi_z$ such that $z_1\cdot z_2\cdot z_3=1$. Furthermore, the distribution $\Pi_{t+1}(\cP)$ assigns equal probabilities to the subspaces $\pi_z $ for each $z\in\{-1,1\}^3$ such that $ z_1\cdot z_2\cdot z_3=1$.  
	 %we observe that the partitions corresponding to $z\in \{-1,1\}^3$ with $z_1z_2z_3 = 1$ have equal measure under $\cP$, and the ones corresponding to $z_1z_2z_3 = -1$ have zero measure under $\cP$.

	Let $i\in [3]$, and $\chi \in \widehat{\cV}\setminus \{\emptyset\}$ be used to partition $\pi$ in the construction of $\Pi_{t+1}$ from $\Pi_t$. Sample a uniformly random $z\in \{-1,1\}^3$ such that $z_1\cdot z_2\cdot z_3=1$ and consider the corresponding subspace $\pi_z$.  For any $j\in [3]$, the quantity $\widehat{E_j|_{ \pi_j}}(\emptyset)$ is updated to either $\widehat{E_j|_{ \pi_j}}(\emptyset) + \widehat{E_j|_{ \pi_j}}(\chi_{a_j})$, or to $\widehat{E_j|_{ \pi_j}}(\emptyset) - \widehat{E_j|_{ \pi_j}}(\chi_{a_j})$, depending on the value of $z_j\in\{-1,1\}$. %\textcolor{red}{(Not saying which of $z_j=1$ or -1 leads to + or - since fourier coefficients may be dependent on choice of vector other than this same $'a'$. Also choice of $a$ to define partitions not unique)}
	Since under the distribution $\Pi_{t+1}(\cP)$, the subspaces $\pi_z$ for $z\in\{-1,1\}^3: z_1\cdot z_2\cdot z_3=1$ are equally likely (in particular $z_j=1$ and $z_j=-1$ are equally likely), this quantity is updated to  $\widehat{E_j|_{ \pi_j}}(\emptyset) + \widehat{E_j|_{ \pi_j}}(\chi_{a_j})$ with probability $\frac{1}{2}$ and $\widehat{E_j|_{ \pi_j}}(\emptyset) - \widehat{E_j|_{ \pi_j}}(\chi_{a_j})$ with probability $\frac{1}{2}$. This implies that the contribution of $\pi_j$ to the change in potential from time $t$ to $t+1$ is exactly
	\[ \frac{1}{2}\left(   \left(\widehat{E_j|_{ \pi_j}}(\emptyset) + \widehat{E_j|_{ \pi_j}}(\chi_{a_j}) \right)^2 +   \left(\widehat{E_j|_{ \pi_j}}(\emptyset) - \widehat{E_j|_{ \pi_j}}(\chi_{a_j}) \right)^2  \right) - \left(\widehat{E_j|_{ \pi_j}}(\emptyset)\right)^2 = \left(\widehat{E_j|_{ \pi_j}}(\chi)\right)^2.\]
	Adding the above equation for all $j\in [3]$ implies that the change in potential due to $\pi$ is  $\sum_{j\in[3]}  \pbra{\widehat{E_j|_{ \pi_j}}(\chi)}^2$. Note that the refinement is performed only when with probability at least $\delta$ over the choice of $\pi \sim \Pi_t(\cP)$, it holds that for some $i\in[3]$, $\abs{\widehat{E_i|_{ \pi_i}}(\chi)}\geq \delta$. This implies that the overall change $\Phi(\Pi_{t+1}) - \Phi(\Pi_t)$ is at least $\delta \cdot \delta ^2 = \delta^3$. This completes the proof.
\end{proof}

\bibliographystyle{alpha}
\bibliography{main}

\end{document}

%% file: macros.tex
\newif\ifnotes
\notesfalse
\newcommand{\authnote}[3]{\textcolor{#3}{[{\footnotesize {\bf #1:} { {#2}}}]}}
\newcommand{\jnote}[1]{\ifnotes \authnote{J}{#1}{purple} \fi}

\newcommand{\unote}[1]{\ifnotes \authnote{U}{#1}{blue} \fi}

\newcommand{\textdef}[1]{\textnormal{\textsf{#1}}}

\newtheorem{theorem}{Theorem}[section]

\newtheorem{proposition}[theorem]{Proposition}
\newtheorem{claim}[theorem]{Claim}
\newtheorem{lemma}[theorem]{Lemma}

\newtheorem{fact}[theorem]{Fact}

\Crefname{importedtheorem}{Imported Theorem}{Imported Theorems}
\Crefname{theorem}{Theorem}{Theorems}
\Crefname{proposition}{Proposition}{Propositions}
\Crefname{claim}{Claim}{Claims}
\Crefname{lemma}{Lemma}{Lemmas}
\Crefname{conjecture}{Conjecture}{Conjectures}
\Crefname{corollary}{Corollary}{Corollaries}
\Crefname{construction}{Construction}{Constructions}
\Crefname{property}{Property}{Properties}

\theoremstyle{definition}
\newtheorem{definition}{Definition}

\Crefname{definition}{Definition}{Definitions}
\Crefname{assumption}{Assumption}{Assumptions}
\Crefname{notation}{Notation}{Notations}

\theoremstyle{remark}

\Crefname{question}{Question}{Questions}
\Crefname{remark}{Remark}{Remarks}
\Crefname{fact}{Fact}{Facts}

\def\cD{{\cal D}}

\def\cF{{\cal F}}
\def\cG{{\cal G}}

\def\cP{{\cal P}}
\def\cQ{{\cal Q}}

\def\cU{{\cal U}}
\def\cV{{\cal V}}

\def\cX{{\cal X}}
\def\cY{{\cal Y}}

\def\bbF{{\mathbb F}}

\def\bbN{{\mathbb N}}

\def\bbR{{\mathbb R}}

\newcommand{\F}{\bbF}

\DeclareMathOperator*{\expectation}{\mathbb{E}}
\newcommand{\E}{\expectation}

\newcommand{\dtv}{d_{\mathsf{TV}}}

\newcommand{\pmone}{\{-1, 1\}}
\newcommand{\eqdef}{\stackrel{\mathsf{def}}{=}}

\newcommand{\wt}{\mathsf{wt}}
\newcommand{\hwt}{\mathsf{hwt}}
\newcommand{\GHZ}{\mathsf{GHZ}}
\newcommand{\val}{\mathsf{val}}
\newcommand{\supp}{\mathsf{supp}}

\newcommand{\eps}{\varepsilon}
\newcommand{\abs}[1]{\left| #1 \right|}

\newcommand{\pbra}[1]{\left( #1 \right)}
\newcommand{\sbra}[1]{\left[ #1 \right]}

\DeclarePairedDelimiterX{\infdivx}[2]{(}{)}{%
	#1\;\delimsize\|\;#2%
}

%% file: intro_2.tex
\section{Introduction} \label{sec:introduction}
The focus of this paper is multi-player games, and in particular their asymptotic behavior under parallel repetition.

Multi-player games consist of a one-round interaction between a referee and $k$ players.  In this interaction, the referee first samples a ``query'' $(q_1, \ldots, q_k)$ from some joint query distribution $\cQ$, and for each $i$ sends $q_i$ to the $i^{th}$ player.  The players are required to respectively produce ``answers'' $a_1, \ldots, a_k$ without communicating with one another (that is, each $a_i$ is a function only of $q_i$) and they are said to \emph{win} the game if $(q_1, \ldots, q_k, a_1, \ldots, a_k)$ satisfy some predicate $W$ that is fixed and associated with the game.

Suppose that a game $G$ has the property that the maximum probability with which players can win is $1 - \epsilon$, no matter what strategy they use.  This quantity is called the \textdef{value of $G$}.  The parallel repetition question \cite{FortnowRS88} asks
\begin{quote}
    \emph{How well can the players concurrently play in $n$ independent copies of $G$?}
\end{quote}

More precisely, consider the following $k$-player game, which we call the \textdef{$n$-wise parallel repetition} of $G$ and denote by $G^n$:
\begin{enumerate}
    \item The referee samples, for each $i \in [n]$ independently, query tuples $(q_1^i, \ldots, q_k^i) \sim \cQ$.  We refer to the index $i$ as a \textdef{coordinate} of the parallel repeated game.
    \item The $j^{th}$ player is given $(q_j^1, \ldots, q_j^n)$ and is required to produce a tuple $(a_j^1, \ldots, a_j^n)$.
    \item The players are said to win in coordinate $i$ if $(q_1^i, \ldots, q_k^i, a_1^i, \ldots, a_k^i)$ satisfies $W$.  They are said to win (without qualification) if they win in every coordinate $i \in [n]$.
\end{enumerate}

One might initially conjecture %\footnote{Indeed, this was asserted by \cite{FortnowRS88}} 
that the value of $G^n$ is $(1 - \epsilon)^n$.   However, this turns out not to be true~\cite{Fortnow89,Feige91,FeigeV02,Raz11}, as players may benefit from correlating their answers across different coordinates.  Still, Raz showed that if $G$ is a two-player game, then the value of $G^n$ is $2^{-\Omega(n)}$, where the $\Omega$ hides a game-dependent constant~\cite{Raz98,Hol09}. Tighter results, based on the value of the initial game are also known \cite{DS14, BG15}. For many applications, such bounds are qualitatively as good as the initial flawed conjecture.

Games involving three or more players have proven more difficult to analyze, and the best known general bound on their parallel repeated value is due to Verbitsky~\cite{Verbitsky94}.  This bound states that the value of $G^n$ approaches $0$, but the bound is \emph{extremely} weak (it shows that the value is at most $\frac{1}{\alpha(n)}$, where $\alpha$ denotes an inverse Ackermann function).  The weakness of this bound is generally conjectured to reflect limitations of current proof techniques rather than a fundamental difference in the behavior of many-player games.  In the technically incomparable but related \emph{no-signaling setting} however, Holmgren and Yang showed that three-player games genuinely behave differently than two-player games~\cite{HolmgrenY19}.  Specifically, they showed that there exists a three-player game with ``no-signaling value'' bounded away from 1 such that no amount of parallel repetition reduces the no-signaling value at all. 

Parallel repetition is a mathematically natural operation that we find worthy of study in its own right.   At the same time, parallel repetition bounds have found several applications  in theoretical computer science (see this survey by \cite{Raz10}).
For example, parallel repetition of 2 player games shares intimate connections with multi-player interactive proofs \cite{BOGKW88}, probabilistically checkable proofs and hardness of approximation \cite{BGS98, Fei98, Has01}, geometry of foams \cite{FKO07, KORW08, AK09}, quantum information \cite{CHTW04}, and communication complexity \cite{PRW97, BBCR13}. Recent work also shows that strong parallel repetition for a particular class of multiprover games implies new time lower bounds on Turing machines that can take advice \cite{MR21}.

Dinur et al. \cite{DHVY17} describe a restricted class of multi-player games for which Raz's approach generalizes (giving exponential parallel bounds).  Specifically, they consider games whose query distribution satisfies a certain connectivity property.  For games outside this class, Verbitsky's bound was the best known.
Dinur et al. highlighted one simple three-player game, called the GHZ game~\cite{GHZ}, that in some sense is maximally far from the aforementioned tractable class of multi-player games.  In the GHZ game, the players' queries are $(q_1, q_2, q_3)$ chosen uniformly at random from $\{0,1\}^3$ such that $q_1 \oplus q_2 \oplus q_3 = 0$, and the players' goal is to produce $(a_1, a_2, a_3)$ such that $a_1 \oplus a_2 \oplus a_3 = q_1 \lor q_2 \lor q_3$.  Dinur et al. conjectured that parallel repetition decreases the value of the GHZ game exponentially quickly, and speculated that progress on proving this would shed light on parallel repetition for general games.
The GHZ game has also played a foundational role in the understanding of quantum information theory, due in part to the fact that quantum strategies can win the GHZ game with probability $1$.  It is possible that improved parallel repetition bounds will find applications in this setting as well. 

In a recent work, Holmgren and Raz~\cite{HR20} proved the following polynomial upper bound on the parallel repetition of the GHZ game:
\begin{theorem}
    The value of the $n$-wise repeated GHZ game is at most $n^{-\Omega(1)}$.
\end{theorem}
Our main contribution is a different proof of this theorem that, in our view, is significantly simpler and more direct than the proof of \cite{HR20}.  Like \cite{HR20}, we actually do not rely on any properties of the GHZ game other than its query distribution, and in particular we do not rely on specifics of the win condition.
Furthermore, unlike most previous works on parallel repetition, our proof makes no use of information theory, and instead relies on the use of Fourier analysis.

\subsection{Technical Overview}
Let $\cP$ denote the distribution of queries in the $n$-wise parallel repeated $\GHZ$ game. Let $\alpha=\Theta(1/n^{\eps})$ for a small constant $\eps>0$ and $E=E_1\times E_2\times E_3$ be any product event with significant probability under $\cP$, i.e., $\cP(E)\ge \alpha$.  The core of our proof is establishing that  for a random coordinate $i \in [n]$, the query distribution $\cP|E$ ($\cP$ conditioned on $E$) is mildly hard in the $i^{th}$ coordinate.  That is, given queries sampled from $\cP|E$, the players' maximum winning probability in the $i^{th}$ coordinate is bounded away from $1$.  Using standard arguments from the parallel repetition literature, this will imply an inverse polynomial bound for the value of the $n$-fold $\GHZ$ game.  The difficulty, as usual, is that the $n$ different queries in $\cP | E$ may not be independent.   

Our approach at a high level is to:
\begin{enumerate}
    \item Identify a class $\cD$ of simple distributions (over queries for the $n$-wise repeated GHZ game) such that it is easy to analyze (in step~\ref{step:most-hard} below) which coordinates are hard for any given $D \in \cD$.  By hard, we mean that the players' maximum winning probability in the $i^{th}$ coordinate is $\frac{3}{4}$.
    \item Approximate $\cP | E$ by a convex combination of distributions from $\cD$.  That is, we write \[
    \cP | E \approx  \sum_j p_j D_j,
    \] where $\{D_j\}$ are distributions in $\cD$, $p_j$ are non-negative reals summing to $1$, and $\approx$ denotes closeness in total variational distance.
    \item \label{step:most-hard} Show that in the above convex combination, ``most'' of the $D_i$ have many hard coordinates.  More precisely, if we sample $j$ with probability $p_j$, then the expected fraction of coordinates in which $D_j$ is hard is at least a constant (say $1/3$).
\end{enumerate}
Completing this approach implies that if $i \in [n]$ is uniformly random, then the $i^{th}$ coordinate of $\cP | E$ can be won with probability at most $1 - \Omega(1)$.  We elaborate on each of these steps below.

\paragraph{Bow Tie Distributions}
For our class of ``simple'' distributions $\cD$, we introduce the notion of a ``bow tie'' distribution.  We then define $\cD$ to be the set of all bow tie distributions.  A \textdef{bow tie} is a set $B$ of the form \[
\left \{ 
\begin{array}{c}
(x_0, y_0, z_0), \\
(x_0, y_1, z_1), \\
(x_1, y_0, z_1), \\
(x_1, y_1, z_0)
\end{array}
\right \} \subseteq (\F_2^n)^3
\]
such that for each $(x,y,z)$ in $B$, we have $x + y + z = 0$.    In particular this requires that $x_0 + x_1 = y_0 + y_1 = z_0 + z_1$.  A \textdef{bow tie distribution} is the uniform distribution on a bow tie.  Our name of ``bow tie'' is based on the fact that bow ties are thus determined by $\{(x_0, y_0), (x_0, y_1), (x_1, y_0), (x_1, y_1) \}$, which we sometimes view as a set of edges in a graph.  In this case, bow ties are special kinds of $K_{2,2}$ subgraphs, where $K_{2,2}$ denotes the complete bipartite graph.

The main property of a bow tie distribution $D$ is that for every coordinate $i$ for which $(x_0)_i \neq (x_1)_i$ (equivalently $(y_0)_i \neq (y_1)_i$, or  $(z_0)_i \neq (z_1)_i$), the $i^{th}$ coordinate of $D$ is as hard as the GHZ game (i.e. players cannot produce winning answers for the $i^{th}$ coordinate with probability more than $\frac{3}{4}$).  This follows by ``locally embedding'' the (unrepeated) GHZ query distribution into the $i^{th}$ coordinate of $D$ as follows.  We first swap $x_0 \leftrightarrow x_1$, $y_0 \leftrightarrow y_1$, $z_0 \leftrightarrow z_1$ as necessary to ensure that \begin{equation}
    \label{eq:standard-form}
    (x_0)_i = (y_0)_i = (z_0)_i = 0.
\end{equation}
An even number of swaps are required to do this by the assumption that $x_0 + y_0 + z_0 = 0$, and bow ties are invariant under an even number of such swaps.   Thus \cref{eq:standard-form} is without loss of generality.  Suppose $\bar{f}_1, \bar{f}_2, \bar{f}_3 : \F_2^n \to \F_2$ comprise a strategy for the $i^{th}$ coordinate of $D$. Then a strategy $f_1, f_2, f_3 : \F_2 \to \F_2$ for the basic (unrepeated) GHZ game can be constructed as
\[
\begin{array}{l}
f_1(b) = \bar{f}_1(x_b) \\
f_2(b) = \bar{f}_2(y_b) \\
f_3(b) = \bar{f}_3(z_b).
\end{array}
\]

The winning probability of this strategy is the same as the winning probability of $\bar{f}_1, \bar{f}_2, \bar{f}_3$ in the $i^{th}$ coordinate because $\big ((x_{b_1})_i, (y_{b_2})_i, (z_{b_3})_i \big) = (b_1, b_2, b_3)$. Hence both probabilities are at most 3/4.

\paragraph{Approximating $\cP | E$ by Bow Ties}
We now sketch how to approximate $\cP | E$ by a convex combination of bow tie distributions, where $E$ is a product event $E_1 \times E_2 \times E_3$.  We assume for now that the non-zero Fourier coefficients of each $E_j$ are small.  We will return to this assumption at the end of the overview --- it turns out to be nearly without loss of generality.

We show that $\cP | E$ is close in total variational distance to the distribution obtained by sampling a \emph{uniformly random} bow tie $B \subseteq E$, and then outputting a random element of $B$.  The latter distribution is equivalent to sampling $(x, y, z)$ with probability proportional to the number of bow ties $B \subseteq E$ that contain $(x, y, z)$.  This number is
\begin{equation}
\label{eq:bowtie-count}
\begin{cases}
\left (\sum_{z' \in \F_2^n} E_1(y + z')  E_2(x + z') E_3(z') \right ) - 1 & \text{if $(x, y, z) \in \supp(\cP | E)$} \\
0 & \text{otherwise,}
\end{cases}
\end{equation}
where we identify $E_1$, $E_2$, and $E_3$ with their indicator functions. Note that we are subtracting $1$ to cancel the term corresponding to $z'=z$.

Intuitively, the fact that all $E_j$ have small Fourier coefficients means that they look random with respect to linear functions.  Thus, one might guess that the above sum is close to $2^n \cdot \mu(E_1) \mu(E_2) \mu(E_3)$ for most $(x, y, z) \in \supp(\cP | E)$, where $\mu(S)=|S|/2^n$ denotes the measure of $S$ under the uniform distribution on $\bbF_2^n$.  If ``close to'' and ``most'' have the right meanings, then this would imply that our distribution is close in total variational distance to $\cP | E$ as desired.

Our full proof indeed establishes this.  More precisely, we view \cref{eq:bowtie-count} as a vector indexed by $(x, y, z)$ and establish bounds on that vector's $\ell_1$ and $\ell_2$ norms as a criterion for near-uniformity.  In the process our proof repeatedly uses the following claims (see \cref{lem:fourierlemmaedge}).  For all  sets $S, T \subseteq \F_2^n$ that are sufficiently large, we have
\[  \E_{\substack{z\sim \bbF_2^n\\x\sim \bbF_2^n}}[S(x)\cdot T(x+z)\cdot E_3(z)] \approx  \mu(S)\cdot \mu(T)\cdot \mu(E_3) \]
and
\[ 
\E_{z\sim \bbF_2^n}\sbra{\pbra{ \E_{x\sim \bbF_2^n}[S(x)\cdot E_2(x+z)]}^2\cdot E_3(z)} \approx  \mu(S)^2\cdot \mu(E_2)^2\cdot \mu(E_3).
\]

\paragraph{Most Bow Ties are Hard in Many Coordinates}

For the final step of our proof, we need to show that the distribution of bow ties analyzed in the previous step produces (with high probability) bow ties that differ in many coordinates.  

We begin by parameterizing a bow tie by $(x_0, y_0, x_0 \oplus x_1)$ and noting that in the previous step, we essentially showed that $E$ contains $2^{3n - O(\log n)}$ different bow ties.  The $O(\log n)$ term in the exponent arises from the fact that the events $\{E_j\}$ have density in $\F_2^n$ that is inverse polynomial in $n$.  A simple counting argument then shows that for a random bow tie, the min-entropy of $x_0 \oplus x_1$ is close to $n$.  This means that $x_0 \oplus x_1$ is close to the uniform distribution in the sense that any event occurring with probability $p$ under the uniform distribution occurs with probability $p \cdot n^{O(1)}$ under the distribution of $x_0 \oplus x_1$.  Thus we can finally apply a Chernoff bound to deduce that with all but $2^{-\Omega(n)}$ probability, $x_0 \oplus x_1$ has Hamming weight at least $n / 3$.

In other words, a bow tie sampled uniformly at random differs in at least a $\frac{1}{3}$ fraction of coordinates.  By the main property of bow ties, this implies that the corresponding bow tie distribution is hard on a $\frac{1}{3}$ fraction of coordinates (indeed, the same set of coordinates).

\paragraph{Handling General Events}
For general (product) events $E = E_1 \times E_2 \times E_3$ (where the sets $\{E_i\}$ need not have small Fourier coefficients), we can partition the universe $\F_2^n \times \F_2^n \times \F_2^n$ into parts $\pi$ such that for most of the parts $\pi$, the event $E$ restricted to $\pi$ has the structure that we already analyzed.  For this to make sense, we ensure several properties of the partition.  First, $\pi$ should be a product set ($\pi = \pi_1 \times \pi_2 \times \pi_3$) so that $E \cap \pi$ is a product set as well, i.e. $E \cap \pi$ has the form $\tilde{E}_1 \times \tilde{E}_2 \times \tilde{E}_3$.  Second, each $\pi_i$ should be an affine subspace of $\F_2^n$ so that we can do Fourier analysis with respect to this subspace.  Finally $\pi_1$, $\pi_2$, and $\pi_3$ should  all be affine shifts of the \emph{same} linear subspace so that the set $\{(x, y, z) \in \pi : x + y + z = 0\}$ has the same Fourier-analytic structure as the parallel repeated GHZ query set $\{(x, y, z) \in (\F_2^{n'})^3 : x + y + z = 0\}$ for some $n' < n$.

We prove the existence of such a partition with $n'$ not too small ($n' = n - o(n)$) by a simple iterative approach, which is similar to \cite{HR20}.

\subsection{Comparison to \cite{HR20}}
Our proof has some similarity to \cite{HR20} --- in particular, both proofs partition $(\F_2^n)^3$ into subspaces according to Fourier-analytic criteria and analyze these subspaces separately --- but the resemblance ends there.  In fact, there are fundamental high-level differences between the two proofs.  

The biggest qualitative difference is that our high-level approach decomposes any conditional distribution $\cP | E$ into components (bow tie distributions) for which many coordinates are hard.  \cite{HR20} takes an analogous approach, but it establishes a weaker result that differs in the order of quantifiers: it first fixes a strategy $f$, and then decomposes $\cP | E$ into components such that $f$ performs poorly on many coordinates of many components.  This difference is due to the fact that \cite{HR20} uses uniform distributions on high-dimensional affine spaces as their basic ``hard'' distributions.  It is not in general possible to express $\cP | E$ as a convex combination of such distributions (for example if each $E_j$ is a uniformly random subset of $\F_2^n$).  Instead, \cite{HR20} expresses $\cP | E$ as a convex combination of ``pseudo-affine'' distributions.  This significantly complicates their proof, and we avoid this complication entirely by our use of bow tie distributions, which are novel to this work.

The remainder of our proof (the analysis of hardness within each part of the partition) is entirely different.

%% file: prelims_2.tex
\section{Notation \& Preliminaries}
 
A significant portion of these preliminaries is taken verbatim from \cite{HR20}.

We write $\exp(t)$ to denote $e^{t}$ for $t\in \bbR$. 

Let $n\in \bbN$. For a vector $v\in\bbR^n$ and $i\in[n]$, we write $v(i)$ or $v^i$ to denote the $i$-th coordinate of $v$. For $p\in\bbN$, we write $\|v\|_p\eqdef\pbra{\sum_{i\in[n]}|v(i)|^p}^{1/p}$ to denote the $\ell_p$ norm of $v$.
For $z\in \{0,1\}^*$,  $\hwt(z)\eqdef\|z\|_1$ denotes the Hamming weight of $z$. 
  
We crucially rely on the Cauchy-Schwarz inequality.
\begin{fact}[Cauchy-Schwarz]
Let $k\in \bbN$ and $a_1,\ldots,a_k,b_1,\ldots,b_k\in \bbR$. Then, 
$\sum_{i=1}^k \abs{ a_i \cdot  b_i }\le \sqrt{\sum_{i=1}^k a_i^2 }\cdot \sqrt{\sum_{i=1}^k b_i^2}$.
\end{fact}

\subsection{Set Theory}

Let $\Omega$ be a universe. By a partition of $\Omega$, we mean a collection of pairwise disjoint subsets of $\Omega$, whose union equals $\Omega$. 
If $\Pi$ is a partition of $\Omega$ and $\omega$ is an element of $\Omega$, we will write $\Pi(\omega)$ to denote the (unique) element of $\Pi$ that contains $\omega$. Thus, we can view $\Pi$ as a function $\Pi : \Omega \to 2^\Omega$.
  
For a set $S\subseteq \Omega$, we identify $S$ with its indicator function $S:\Omega\to\{0,1\}$ defined at $\omega\in \Omega$ by 
\[ S(\omega)=\begin{cases} 1 & \text{if }\omega\in S\\ 0 & \text{otherwise.}\end{cases} \] 
  
For sets $S,T\subseteq \Omega$ such that $T\neq \emptyset$, we use $S|_T\subseteq T$ to denote the set $S\cap T$ when viewed as a subset of $T$. In particular, $S|_T$ is an indicator function from $T$ to $\{0,1\}$.
  
\subsection{Probability Theory}
\label{sec:probability}

\paragraph*{Probability Distributions.} Let $P$ be a distribution over a universe $\Omega$. We sometimes think of $P$ as a vector in $\bbR^{|\Omega|}$ whose value in coordinate $\omega\in \Omega$ is $P(\omega)$. In particular, we use $\|P-Q\|_1$ to denote the $\ell_1$ norm of the vector $P-Q \in \bbR^{|\Omega|}$, where $P$ and $Q$ are probability distributions. We use $\omega\sim P$ to denote a random element $\omega$ distributed according to $P$. We use $\supp(P) = \{\omega \in \Omega : P(\omega) > 0\}$ to denote the \textdef{support} of the distribution $P$. 
  
\paragraph*{Random Variables} Let $\Sigma$ be any alphabet. We say that $X : \Omega \to \Sigma$ is a $\Sigma$-valued random variable. If $\Sigma=\bbR$, we say that the random variable is real-valued. If $X$ is a real-valued random variable, the \textdef{expectation of $X$
    under $P$} is denoted $\E_{\omega \sim P}[X(\omega)]$. Often, the underlying distribution $P$ is implicit, in which case we simply use $\E[X]$. If $X$ is a $\Sigma$-valued random variable and $P$ is a probability distribution, we write
$P_X$ or $X(P)$ to denote the induced probability distribution of $X$ under $P$, i.e., $P_X(\sigma)=(X(P))(\sigma) \eqdef
P(X = \sigma)$ for all $\sigma\in \Sigma$. In particular, we say that $X$ is distributed according to $P_X$ and we use $\sigma \sim X(P)$ to denote a random variable $\sigma$ distributed according to $P_X$. The distribution $P$ is often implicit, and we identify $X$ with the underlying distribution $P_X$.  
 
\paragraph*{Events.}  We refer to subsets of $\Omega$ as \textdef{events}.  We use standard shorthand
for denoting events.  For instance, if $X$ is a $\Sigma$-valued random variable
and $x \in \Sigma$, we write $X = x$ to denote the event
$\{\omega \in \Omega : X(\omega) = x\}$. Similarly, for a subset $F\subseteq \Sigma$, we write $X\in F$ to denote the event $\{ \omega\in \Omega : X(\omega)\in F\}$.
We use $P(E)$ to denote the probability of $E$ under $P$. When $P$ is implicit, we use the notation $\Pr(E)$ to denote $P(E)$. 

\paragraph*{Conditional Probabilities}
  Let $E \subseteq \Omega$ be an event with $P(E) > 0$. Then the
  \textdef{conditional distribution of $P$ given $E$} is denoted
  $(P | E) : \Omega \to \bbR$ and is defined to be
\[
  (P|E)(\omega) = \begin{cases}
    P(\omega) / P(E) & \text{if $\omega \in E$} \\
    0 & \text{otherwise.}
  \end{cases}
\]
If $E$ is an event, we write $P_{X|E}$ as shorthand for $(P | E)_X$.

\paragraph*{Measure under Uniform Distribution}\label{sec:measure}
For any set $S\subseteq \Omega$, we sometimes identify $S$ with the uniform distribution over $S$. In particular, we use $x\sim S$ to denote $x$ sampled according to the uniform distribution on $S$. For $S,\pi\subseteq \Omega$ such that $\pi\neq \emptyset$, we use $\mu_\pi(S)=\frac{|S\cap \pi|}{|\pi|}$ to denote the measure of $S$ under the uniform distribution over $\pi$. When $\pi=\Omega$, we omit the subscript and simply use $\mu(S)$.

\subsection{Fourier Analysis}
\label{sec:fourier}
\paragraph*{Fourier Analysis over Subspaces} For any (finite) vector space $\cV$ over $\F_2$, the \textdef{character group of
  $\cV$}, denoted $\widehat{\cV}$, is the set of group homomorphisms mapping $\cV$ (viewed
as an additive group) to $\pmone$ (viewed as a multiplicative group).  Each such
homomorphism is called a \textdef{character} of $\cV$. For functions mapping $\cV \to \bbR$, we define the inner product
\[
  \langle f, g \rangle \eqdef \E_{x \sim \cV} \left [ f(x) {g(x)} \right ].
\] 
The character group of $\cV$ forms an orthonormal basis under this inner product. We refer to the all-ones functions $\chi:\cV\to \{-1,1\}, \chi\equiv 1$ as the \textdef{trivial character} or the \textdef{zero character} and denote this by $\chi=\emptyset$. 

For all characters $\chi\neq \emptyset$, since $\langle \chi, \emptyset \rangle =0$, we have $\E_{x\sim \cV}\sbra{\chi(x)}=0$, in particular, $\chi(\cV)$ is a uniform $\{\pm1\}$-random variable. Let $\emptyset\neq S\subseteq \cV$ be a set. Then $\mu_{\cV}(S)\triangleq \frac{|S\cap \cV|}{|\cV|}=\widehat{S}(\emptyset)$, where we identify $S$ with its indicator function $S:\cV\to \{0,1\}$ as mentioned before. For $\chi\in \widehat{\cV}$, we have $\E_{x\sim S}\sbra{\chi(x)}=\tfrac{\widehat{S}(\chi)}{\widehat{S}(\emptyset)}$.
\begin{fact}
  Given a choice of basis for $\cV$, there is a canonical isomorphism between $\cV$
  and $\widehat{\cV}$.  Specifically, if $\cV = \F_2^n$, then the characters of $\cV$ are
  the functions of the form
  \[
    \chi_\gamma(v) = (-1)^{\gamma \cdot v}
  \]
  for $\gamma \in \F_2^n$.
\end{fact}

\begin{definition}
For any function $f : \cV \to \bbR$, its \textdef{Fourier transform} is the function
$\widehat{f} : \widehat{\cV} \to \bbR$ defined by
\[
  \widehat{f}(\chi) \eqdef \langle f, \chi \rangle = \E_{x \sim \cV} \left [ f(x) \chi(x) \right ].
\]
\end{definition}
Since the characters of $\cV$ are orthonormal and $\cV$ is finite, we can deduce that $f$ is equal to
$\sum_{\chi \in \widehat{\cV}} \widehat{f}(\chi) \cdot \chi$. 

\begin{theorem}[Plancherel]
  \label{thm:plancherel}
  For any $f, g : \cV \to \bbR$,
  \[
    \langle f, g \rangle =\sum_{\chi\in \widehat{\cV}}\widehat{f}(\chi)\cdot \widehat{g}(\chi).
  \]
\end{theorem}
An important special case of Plancherel's theorem is Parseval's theorem:
\begin{theorem}[Parseval]
  \label{thm:parseval}
  For any $f : \cV \to \bbR$,
  \[
    \E_{x\sim \cV}\sbra{f(x)^2} =\sum_{\chi\in \widehat{\cV}} \widehat{f}(\chi)^2 .
  \]
\end{theorem}

\paragraph*{Fourier Analysis over Affine Subspaces}

Fix any subspace $\cV\subseteq \bbF_2^n$ and a vector $a\in \bbF_2^n$. Let $\cU=a+\cV$ denote the affine subspace obtained by shifting $\cV$ by $a$. For every function $f:\cV\to \bbR$, we associate it with a function $f_a:\cU\to \bbR$ defined by $f_a(x)=f(x+a)$ for all $x\in \cU$. This is a bijective correspondence between the set of functions from $\cU$ to $\bbR$ and the set of functions from $\cV$ to $\bbR$. Under this association, we can identify $\chi\in \widehat{\cV}$ with $\chi_a:\cU \to \{-1,1\}$ where $\chi_a(x)=\chi(x+a)$ for all $x\in \cU$. This defines an orthonormal basis  $\widehat{\cU}_a:= \{\chi_a: \cU\to \{-1,1\}\, | \,\chi\in \widehat{\cV}\}$ for the vector space of functions from $\cU$ to $\bbR$. We call this the Fourier basis for $\cU$ with respect to $a$. This basis depends on the choice of the shift $a\in \cU$. However, for all possible shifts $b\in \cU$ and character functions $\chi\in \widehat{\cV}$, the functions $\chi_a$ and $\chi_b$ only differ by a sign. To see this, observe that 
\[ \chi_a(x)=\chi(a+x)=\chi(b+x)\cdot \chi(a+b)=\chi_b(x)\cdot \chi(a+b) \]
We will sometimes ignore the subscript and simply use $\chi\in \widehat{\cV}$ to index functions in the Fourier basis of $\cU$. This is particularly the case when the properties we are dealing are independent of choice of basis (for example, the absolute values of Fourier coefficients of a function).

\subsection{Multi-Player Games}
In parallel repetition we often work with Cartesian product sets of the form
$(\cX_1 \times \cdots \times \cX_k)^n$.  For these sets, we will use subscripts to index the inner
product and superscripts to index the outer product.  That is, for $\cX=\cX_1\times \ldots \times \cX_k$ we view elements $x$ of $\cX^n$ as tuples $(x_1, \ldots, x_k)$, where $x_i\in \cX_i^n$. We use $x_i^j$ or $x_i(j)$ to refer to the $j^{th}$ coordinate of $x_i$. We use $x^j$ to denote the vector $(x_1^j, \ldots, x_k^j)$.  

If
$\{E_i \subseteq \cX_i\}_{i \in [k]}$ is a collection of subsets, we write $E_1 \times \cdots \times E_k$  
to denote the set $\{x \in \cX :\, \forall i \in [k], x_i \in E_i\}$.
We
say that $f :(\cX_1 \times \cdots \times \cX_k)^n\to (\cY_1 \times \cdots \times \cY_k)^n $ is a product function if $f=f_1 \times \cdots \times f_k$ for some functions $f_i:\cX_i^n\to \cY_i^n$.

\begin{definition}[Multi-player Games]
	A \textdef{$k$-player game} is a tuple $(\cX, \cY, Q, W)$, where
	$\cX = \cX_1 \times \cdots \times \cX_k$ and
	$\cY = \cY_1 \times \cdots \times \cY_k$ are finite sets, $Q$ is a probability
	measure on $\cX$, and $W : \cX \times \cY \to \{0,1\}$ is a ``winning'' predicate.  We refer to $Q$ as the \textdef{query distribution} or the \textdef{input distribution} of the game. 
\end{definition}

\begin{definition}[Deterministic Strategies]
	\label{def:deterministic-strategies}
	A \textdef{deterministic strategy} for a $k$-player game $\cG = (\cX, \cY, Q, W)$ is a function $f = f_1 \times \cdots \times f_k $ where each $f_i: \cX_i \to \cY_i$.  The \textdef{success probability} of $f$ in $\cG$ is denoted and defined as
	\[
	\val(\cG,f) \eqdef \Pr_{x \sim Q} \Big [ W \big (x, f(x)\big ) = 1 \Big ].
	\]
\end{definition}

% \begin{definition}
%   The \textdef{success probability} of a function
%   $f = f_1 \times \cdots f_k : \cX \to \cY$ in a $k$-player game $\cG = (\cX, \cY, Q,
%   W)$ is
%   \[
%     v[f](\cG) \eqdef \Pr_{x \sim Q} \Big [ W \big ((x, f(x)\big ) = 1 \Big ].
%   \]
% \end{definition}
The most important quantity associated with a game is the maximum probability with which the game can be ``won''.
\begin{definition}
	\label{def:game-value}
	The \textdef{value} of a $k$-player game
	$\cG = (\cX, \cY, Q, W)$, denoted $\val(\cG)$, is the maximum, over all 
	deterministic strategies $f$, of $\val(\cG,f)$.
\end{definition}

It is often easier to construct \emph{probabilistic} strategies for a
game, i.e. strategies in which players may use shared and/or
individual randomness in computing their answers.

\begin{definition}[Probabilistic Strategies]
	\label{def:probabilistic-strategies}
	Let $\cG = (\cX, \cY, Q, W)$ be a $k$-player game. A \textdef{probablistic strategy} for $\cG$ is a distribution $\cF$ of deterministic strategies for $\cG$.  The \textdef{success probability} of $\cF$ in $\cG$ is denoted and defined as
	\[
	\val(\cG,\cF) \eqdef \Pr_{\substack{x \sim Q \\ f \sim \cF}} \Big [ W \big (x, f(x)\big ) = 1 \Big ].
	\]
\end{definition}

A standard averaging argument implies that for every game,
probabilistic strategies cannot achieve better success probability
than deterministic strategies:
\begin{fact}
	\label{fact:randomized-value}
	Replacing ``deterministic strategies'' by ``probabilistic
	strategies'' in \cref{def:game-value} yields an equivalent
	definition.
\end{fact}

The main operation on multi-player games that we consider in this
paper is parallel repetition:
\begin{definition}[Parallel Repetition]
	Given a $k$-player game $\cG = (\cX, \cY, Q, W)$, its \textdef{$n$-fold parallel
		repetition}, denoted $\cG^n$, is defined as the $k$-player game
	$(\cX^n, \cY^n, Q^n, W^n)$, where $W^n(x,y) \eqdef \bigwedge_{j=1}^{n} W(x^{j}, y^{j})$. For $x\in \cX^n$, we refer to $x_i\in \cX_i^n$ as the input to the $i$-th player.
	% In the above we write elements $x \in \cX'$ as $\left(\{x_1^{(i)}\}_{i \in [n]}, \ldots, \{x_k^{(i)}\}_{i \in [n]}\right)$,  we write
	% $x_j$ to denote $(x_j^{(1)}, \ldots, x_j^{(n)})$,
	% and we write
	% $x^{(i)}$ to denote $(x_1^{(i)}, \ldots, x_k^{(i)})$.
	% Our notation for components of elements of $\cY'$ is analogous.
\end{definition}

To bound the value of parallel repeated games, it is helpful to
analyze the probability of winning in a particular instance of the
game under various modified query distributions.
\begin{definition}[Value in $j^{th}$ coordinate]
	If $\cG = (\cX, \cY, Q, W^n)$ is a game (with a product winning predicate),
	the \textdef{value of $\cG$ in the $j^{th}$ coordinate} for $j\in[n]$, denoted $\val^{(j)}(\cG)$,
	is the value of the game $(\cX, \cY, Q, W')$, where $W'(x, y) = W(x^j, y^j)$.
\end{definition}

\begin{definition}[Game with Modified Query Distribution]
	Let $\cG = (\cX, \cY, Q, W)$ be a game. For a probability measure $P$ on $\cX$, we write
	$\cG | P$ to denote the game $(\cX, \cY, P, W)$.  For an event $E$ on $\cX$, we write
	$\cG | E$ to denote the game $(\cX, \cY, Q_E, W)$.
\end{definition}

\subsection{GHZ Distribution} 

Let $\cX=  \cX_1\times \cX_2\times \cX_3$ and $\cY=\cY_1\times \cY_2\times \cY_3$ where $\cX_i=\cY_i=\bbF_2$. Let $\cQ$ denote the uniform distribution over $\{(0,0,0),(0,1,1),(1,0,1),(1,1,0)\}$. Define $W:\cX\times \cY\to \{0,1\}$ at $x\in \cX,y\in \cY$ by $W(x,y)=1$ if and only if $x_1\lor x_2\lor x_3=y_1+ y_2+ y_3 \pmod{2}$. The $\GHZ$ game refers to the 3-player game $(\cX,\cY,\cQ,W)$, which has value $3/4$.   The $n$-fold repeated $\GHZ$ game refers to the $n$-fold parallel repetition of $(\cX,\cY,\cQ,W)$.  Our parallel repetition results easily generalize with any other (constant-sized) answer alphabet $\cY'$ and any predicate $W'$, as long as the game $(\cX, \cY', \cQ, W')$ has value less than $1$.

We typically use $X=(X_1,X_2,X_3)\in \cX^n$ to denote a random variable distributed according to $\cQ^n$ where $X_i\in \cX_i^n$ denotes the input to the $i$-th player.

%% file: main.bbl
\begin{thebibliography}{BOGKW88}

\bibitem[AK09]{AK09}
Noga Alon and Bo'az Klartag.
\newblock Economical toric spines via {C}heeger's inequality.
\newblock {\em J. Topol. Anal.}, 1(2):101--111, 2009.

\bibitem[BBCR13]{BBCR13}
Boaz Barak, Mark Braverman, Xi~Chen, and Anup Rao.
\newblock How to compress interactive communication.
\newblock {\em SIAM J. Comput.}, 42(3):1327--1363, 2013.
\newblock (also in STOC 2010).

\bibitem[BG15]{BG15}
Mark Braverman and Ankit Garg.
\newblock Small value parallel repetition for general games.
\newblock In {\em STOC}, pages 335--340, 2015.

\bibitem[BGS98]{BGS98}
Mihir Bellare, Oded Goldreich, and Madhu Sudan.
\newblock Free bits, {PCP}s, and nonapproximability---towards tight results.
\newblock {\em SIAM J. Comput.}, 27(3):804--915, 1998.
\newblock (also in FOCS 1995).

\bibitem[BOGKW88]{BOGKW88}
Michael Ben-Or, Shafi Goldwasser, Joe Kilian, and Avi Wigderson.
\newblock Multi-prover interactive proofs: How to remove intractability
  assumptions.
\newblock In {\em STOC}, pages 113--131, 1988.

\bibitem[CHTW04]{CHTW04}
Richard Cleve, Peter H{\o}yer, Benjamin Toner, and John Watrous.
\newblock Consequences and limits of nonlocal strategies.
\newblock In {\em CCC}, pages 236--249, 2004.

\bibitem[DHVY17]{DHVY17}
Irit Dinur, Prahladh Harsha, Rakesh Venkat, and Henry Yuen.
\newblock Multiplayer parallel repetition for expanding games.
\newblock In {\em ITCS}, volume~67 of {\em LIPIcs}, pages Art. No. 37, 16,
  2017.

\bibitem[DS14]{DS14}
Irit Dinur and David Steurer.
\newblock Analytical approach to parallel repetition.
\newblock In {\em STOC}, pages 624--633, 2014.

\bibitem[Fei91]{Feige91}
Uriel Feige.
\newblock On the success probability of the two provers in one-round proof
  systems.
\newblock In {\em CCC}, pages 116--123. {IEEE} Computer Society, 1991.

\bibitem[Fei98]{Fei98}
Uriel Feige.
\newblock A threshold of {$\ln n$} for approximating set cover.
\newblock {\em J. ACM}, 45(4):634--652, 1998.
\newblock (also in STOC 1996).

\bibitem[FKO07]{FKO07}
Uriel Feige, Guy Kindler, and Ryan O'Donnell.
\newblock Understanding parallel repetition requires understanding foams.
\newblock In {\em CCC}, pages 179--192, 2007.

\bibitem[For89]{Fortnow89}
Lance~Jeremy Fortnow.
\newblock {\em Complexity-theoretic aspects of interactive proof systems}.
\newblock PhD thesis, MIT, 1989.

\bibitem[FRS88]{FortnowRS88}
Lance Fortnow, John Rompel, and Michael Sipser.
\newblock On the power of multi-power interactive protocols.
\newblock In {\em CCC}, pages 156--161. {IEEE} Computer Society, 1988.

\bibitem[FV02]{FeigeV02}
Uriel Feige and Oleg Verbitsky.
\newblock Error reduction by parallel repetition - {A} negative result.
\newblock {\em Comb.}, 22(4):461--478, 2002.

\bibitem[GHZ89]{GHZ}
Daniel~M. Greenberger, Michael~A. Horne, and Anton Zeilinger.
\newblock {\em Going Beyond Bell's Theorem}, pages 69--72.
\newblock Springer Netherlands, Dordrecht, 1989.

\bibitem[H{\aa}s01]{Has01}
Johan H{\aa}stad.
\newblock Some optimal inapproximability results.
\newblock {\em J. ACM}, 48(4):798--859, 2001.
\newblock (also in STOC 1997).

\bibitem[Hol09]{Hol09}
Thomas Holenstein.
\newblock Parallel repetition: simplifications and the no-signaling case.
\newblock {\em Theory Comput.}, 5:141--172, 2009.
\newblock (also in STOC 2007).

\bibitem[HR20]{HR20}
Justin Holmgren and Ran Raz.
\newblock A parallel repetition theorem for the {GHZ} game.
\newblock {\em CoRR}, abs/2008.05059, 2020.

\bibitem[HY19]{HolmgrenY19}
Justin Holmgren and Lisa Yang.
\newblock The parallel repetition of non-signaling games: counterexamples and
  dichotomy.
\newblock In {\em {STOC}}, pages 185--192. {ACM}, 2019.

\bibitem[KORW08]{KORW08}
Guy Kindler, Ryan O'Donnell, Anup Rao, and Avi Wigderson.
\newblock Spherical cubes and rounding in high dimensions.
\newblock In {\em FOCS}, pages 189--198, 2008.

\bibitem[MR21]{MR21}
Kunal Mittal and Ran Raz.
\newblock Block rigidity: Strong multiplayer parallel repetition implies
  super-linear lower bounds for turing machines.
\newblock In {\em ITCS}, volume 185 of {\em LIPIcs}, pages 71:1--71:15, 2021.

\bibitem[PRW97]{PRW97}
Itzhak Parnafes, Ran Raz, and Avi Wigderson.
\newblock Direct product results and the {GCD} problem, in old and new
  communication models.
\newblock In {\em STOC}, pages 363--372. 1997.

\bibitem[Raz98]{Raz98}
Ran Raz.
\newblock A parallel repetition theorem.
\newblock {\em SIAM J. Comput.}, 27(3):763--803, 1998.
\newblock (also in STOC 1995).

\bibitem[Raz10]{Raz10}
Ran Raz.
\newblock Parallel repetition of two prover games.
\newblock In {\em CCC}, pages 3--6. 2010.

\bibitem[Raz11]{Raz11}
Ran Raz.
\newblock A counterexample to strong parallel repetition.
\newblock {\em {SIAM} J. Comput.}, 40(3):771--777, 2011.

\bibitem[Ver94]{Verbitsky94}
Oleg Verbitsky.
\newblock Towards the parallel repetition conjecture.
\newblock In {\em CCC}, pages 304--307. {IEEE} Computer Society, 1994.

\end{thebibliography}
